\documentclass[11pt,a4paper]{amsart}
\usepackage[a4paper,margin=25mm]{geometry}

\usepackage{amssymb}
\usepackage{graphicx}
\usepackage{xcolor}
\usepackage{booktabs}
\usepackage{capt-of}
\usepackage[numbers,square]{natbib}
\usepackage[colorlinks=true,linkcolor=blue!55!black,
  citecolor=blue!55!black,urlcolor=blue!55!black]{hyperref}

\makeatletter
\def\paragraph{\@startsection{paragraph}{4}%
  \z@\z@{-\fontdimen2\font}%
  {\normalfont\bfseries}}
\makeatother

\newcommand{\extversion}{short}
\usepackage{todonotes}
\usepackage{mathtools}
\usepackage{etoolbox}

\usepackage{csquotes}

\usepackage{mathpartir}
\usepackage{mdframed}

\usepackage{longtable}
\usepackage{array}

\usepackage{enumitem}
\usepackage{footmisc}

\newcolumntype{P}[1]{>{\raggedright\arraybackslash}p{#1}}

\usepackage{tikz-cd}
\usepackage{wrapfig}
\usepackage{fontawesome5}

\newcommand{\agdabaseurl}{https://nicolaikraus.github.io/extensiontypes-agda/html/index.html}
\newcommand{\formalized}{{\color{blue!55!black}{\raisebox{-0.5pt}{\scalebox{0.8}{\faCog}}}}}
\newcommand{\flinkurl}[1]{\href{#1}{\formalized}}
\newcommand{\flink}[1]{\flinkurl{\agdabaseurl\##1}}

\usepackage{comment}
\providecommand{\extversion}{short}
\newif\ifShorten
\ifdefstring{\extversion}{long}{\Shortenfalse}{\Shortentrue}
\ifShorten
  \excludecomment{longversion}
  \includecomment{shortversion}
\else
  \includecomment{longversion}
  \excludecomment{shortversion}
\fi

\usepackage{microtype}
\DisableLigatures[-]{family=tt*}

\usepackage[capitalize,noabbrev]{cleveref}

\newcommand{\Cofib}{\mathsf{Cofib}}
\newcommand{\Prop}{\mathsf{Prop}}
\newcommand{\UU}{\mathcal{U}}
\newcommand{\UUs}{\mathcal{U}^\textup{s}}

\newcommand{\colonequiv}{\mathrel{\vcentcolon\mspace{-1mu}\equiv}}
\newcommand{\defeq}{\colonequiv}

\newcommand{\cofib}{\hookrightarrow}

\newcommand{\zero}{\mathbf{0}}
\newcommand{\szero}{\mathbf{0}^\textup{s}}
\newcommand{\one}{\mathbf{1}}

\DeclareMathOperator{\id}{id}
\newcommand{\seq}{\mathrel{=^\textup{s}}}
\newcommand{\siso}{\mathrel{\simeq^\textup{s}}}

\newcommand{\poprod}{\mathbin{\widehat{\times}}}

\newcommand{\hfib}[1]{{#1}^{-1, \textup{h}}}
\newcommand{\sfib}[1]{{#1}^{-1, \textup{s}}}

\newcommand{\Ext}[3]{\left\langle #2 \middle|^{#1}_{#3} \right\rangle}

\DeclareMathAlphabet{\mathbbx}{U}{bboldx}{m}{n}

\newcommand{\bbMLTT}{\mathbbx{MLTT}}
\newcommand{\bbHoTT}{\mathbbx{HoTT}}
\newcommand{\bbCub}{\mathbbx{Cub}}

\newcommand{\extpapertitle}{Extension Types for Free}
\newcommand{\extpapersubtitle}{Cofibrations, Gluing, and Univalence in Two-Level Type Theory}

\newcommand{\Glue}{\mathsf{Glue}}
\newcommand{\glue}{\mathsf{glue}}
\newcommand{\unglue}{\mathsf{unglue}}
\newcommand{\rulename}[1]{\textup{\textsc{#1}}}

\theoremstyle{plain}
\newtheorem{theorem}{Theorem}[section]
\newtheorem{lemma}[theorem]{Lemma}
\newtheorem{proposition}[theorem]{Proposition}
\newtheorem{corollary}[theorem]{Corollary}
\newtheorem{conjecture}[theorem]{Conjecture}
\newtheorem*{lemma*}{Lemma}
\theoremstyle{definition}
\newtheorem{definition}[theorem]{Definition}

\theoremstyle{remark}
\newtheorem{remark}[theorem]{Remark}

\title{\extpapertitle: \\[.1cm] \extpapersubtitle}
\newcommand{\extpaperauthor}{Nicolai Kraus}
\author{\extpaperauthor}
\date{}
\hypersetup{%
  pdftitle={\extpapertitle: \extpapersubtitle},%
  pdfauthor={\extpaperauthor}%
}

\begin{document}
\begin{abstract}
\emph{Extension types} are a concept in dependent type theory that has appeared in various contexts.
The idea is to have types whose terms are partially determined, e.g.\ via a strict boundary condition.
Standard examples are path types of cubical type theories (paths with fixed endpoints), Riehl and Shulman's name-giving extension types (terms fixed on subshapes), as well as the controlled-unfolding mechanism of \texttt{cooltt} and \texttt{Agda} (terms that are fixed if a condition is met).
In each case, the type theory is equipped with a (meta-theoretic) \emph{face calculus}, or \emph{shape layer}, that governs their rules, and comes with intended semantics.

We unify all these occurrences in a single framework where no new axioms or model constructions are needed: Extension types come ``for free''.
The framework is two-level type
theory (2LTT), which enriches homotopy type theory (HoTT) with a layer of
strict equality. This step, too, is free (semantically): the standard models of
HoTT are automatically models of 2LTT, and the theory is conservative over HoTT.
Extension types are definable
rather than postulated, and the definition validates Riehl and Shulman's
entire extension-type calculus: the rules hold strictly, and the postulated axioms,
such as relative function extensionality, become theorems.
In this way, every model of the base theory (HoTT) gives rise to a model
of the same theory with extension types; the only genuine assumptions are
which maps count as cofibrations.

Conservativity makes the framework a tool for comparing type theories. We
prove that cubical gluing, in a suitable formulation, is equivalent to
univalence; in particular, this reverses the familiar implication from Glue
types to the univalence axiom. On this basis, we suggest an approach toward
the conjecture that cubical type
theories are conservative over book HoTT, one of the
central open problems of homotopy type theory.

All results of the main body of the paper are
auto-formalized in Agda \texttt{--two-level}, in a development that combines HoTT-internal arguments with reasoning that is external to HoTT.

\end{abstract}

\maketitle

\begin{longversion}
\tableofcontents
\end{longversion}

\section{Introduction}

\paragraph{Extension types.}
Extension types were introduced, under this name, by Riehl and Shulman
\citep[\S2.2]{RS2017}. The informal idea is simple: given a term that is only
defined on a \emph{subshape} of some shape, there is a \emph{type} of all the
ways of extending it to the whole shape. Cubical type theory provides the
basic example: if a type $A$ lives over an interval and a term $a$ of $A$ is
given only over the two endpoints, then the type of extensions of $a$ to the
whole interval is precisely a type of \emph{cubical paths}.

The cubical literature
(e.g.\ \citep{CCHM2018}) writes such types as $A[\varphi \mapsto u]$, where the face
formula $\varphi$ describes the part of a cube on which the partial element
$u$ is defined. Riehl and Shulman's own extension types, in their simplicial
type theory, generalize this from interval endpoints to arbitrary shape
inclusions. A further family of examples is \emph{proposition-indexed}: the
type $\{A \mid \varphi \mapsto a\}$ of elements of $A$ that coincide with a
given $a$ whenever the proposition $\varphi$ holds. This connective underlies
controlled unfolding \citep{gratzer2022controlling} and extent types
\citep{sterling2022first,sterling2021logical}; Zhang \citep{zhang2023three}
gives an overview and discusses a general framework. Note that the cubical
$A[\varphi\mapsto u]$ is itself an instance of this proposition-indexed
reading.

The first goal of this paper is to unify these examples. We take the position
that they are conceptually the same construction: each setting fixes a
\emph{face calculus}, a logical layer whose formulas describe boundaries
(that is, cofibrations), and this calculus determines which extension types
can be formed. The settings differ mainly in the expressiveness of their face
calculus. To capture all of them at once, we therefore choose the most
expressive boundary language available to us: type theory itself. A framework in
which this idea works out of the box is \emph{two-level type theory}
\citep{AltenkirchCapriottiKraus2016,2LTT}, abbreviated 2LTT; as we explain next, it is
semantically \emph{free}. Our goals are then to demonstrate that, in this free
framework, extension types and their entire theory are free and automatic
(the natural definition validates the rules and axioms that Riehl and Shulman
propose), and that the framework is strong enough to prove meta-theoretic
statements. In particular, we use it to establish a precise connection
between cubical \emph{gluing} and \emph{univalence}, which opens a route
toward a conservativity result.

\paragraph{Two-level type theory, and how it is free.}
2LTT combines two layers in a
single language: an \emph{inner} level, which is homotopy type theory (HoTT),
and an \emph{outer} (also: strict, or exo-) level, a type theory with a strict
equality $\seq$ satisfying UIP and function extensionality; on paper, one
may even take the outer level to be extensional Martin-L\"of type theory.
Every inner type is in particular an outer type, but not conversely; the
inner types and their families are called \emph{fibrant}. The outer level is
where our cofibrations live: it is the face calculus.

2LTT is essentially free in terms of semantics. The standard models of HoTT,
such as simplicial and cubical sets, are automatically models of 2LTT, and
more generally, every CwF-style model of HoTT gives rise to a model of 2LTT
via a presheaf construction \citep{2LTT,Capriotti2017}. The semantics moreover
extends to all Grothendieck $(\infty,1)$-toposes: Uskuplu \citep[\S4--\S5]{Uskuplu2025} shows that the
presheaf construction applies to the fibration CwF of any good model category
in the sense of \citet{LumsdaineShulman2020}, and by Shulman's theorem
\citep{Shulman2019Toposes} such model categories present all
$(\infty,1)$-toposes; as Kolomatskaia and Shulman summarize, the semantics of
2LTT is \enquote{not significantly less general than ordinary HoTT}
\citep{KolomatskaiaShulman2025}.

This freeness is a statement about semantics, not about syntax: a proof
assistant implementing HoTT does not automatically implement 2LTT, although
Agda's \texttt{--two-level} mode \citep{agdareadthedocs} exists; as with
any formal system, expressiveness trades against computational behavior.
Since our face calculus is the full outer type theory, its equality is
undecidable, and we give up the decidable boundary checking that CCHM's face
lattice and Riehl--Shulman's tope logic are designed to provide. We regard
this as a parameter rather than a defect: one can imagine defining the
framework relative to a chosen sub-universe $\Cofib$ of strict propositions
closed under the relevant operations (essentially a \emph{dominance} in the
sense of synthetic topology \citep{Rosolini1986,Escardo2004}, or a cofibration
classifier in the style of
Orton--Pitts \citep{OrtonPitts2018}), with the expectation that decidable
face calculi arise as instances. Making this parametrization precise is left
to future work; in the present paper, decidability plays no role.

\paragraph{Extension types for free.}
In 2LTT, the notion of extension type can be defined in a natural way, and all the rules one may usually put on top of the construction are provable.
Given a cofibration $i:\Phi\cofib\Psi$, a fibrant family $A$ over $\Psi$, and
a partial section $a:\prod_{\phi:\Phi}A(i\phi)$, the extension type is the
strict fiber of the restriction map $f \mapsto f \circ i$ over $a$, written
in the angle-bracket notation of Riehl and Shulman \citep{RS2017}:
\[
	\Ext{i}{\Pi_\Psi A}{a}\;\colonequiv\;
	\Sigma\bigl(f:\textstyle\prod_{\psi:\Psi}A(\psi)\bigr).\,(f\circ i\seq a),
\]
cf.\ \cref{def:ext-formal}. With this definition, the rules that Riehl and
Shulman postulate (introduction, elimination, $\beta$, $\eta$, and
the boundary rule) hold strictly, and their axioms and structural laws become provable theorems;
this includes statements such as relative function extensionality, an axiom in \citet{RS2017} but automatic in this
framework,
and the homotopy extension property (a realignment principle), all of which we have formalized in Agda in higher generality than the original paper suggests.
The caveats are two hypotheses on the outer level and one on the cofibration class.
The union law requires the cofibration class to be closed under binary meets: simplicial and cubical face formulas
have this closure property, arbitrary cofibrations do not; therefore, while this property still holds automatically in the situation analogous to what Riehl--Shulman consider, its generality is reduced compared to the rest of the development (see \cref{lem:rs45}).
Beyond that, \cref{lem:rs42} assumes that the pushout-product exists in the outer level and \cref{lem:rs45} that $\Phi\vee\Psi$ is an outer pushout; both are stated in the lemmas and parametrized in the formalization.

While the framework is free, its instances are not. That a given map \emph{is} a
cofibration remains an assumption, to be justified by inspecting the intended
models. 

Because extension types are definable rather than postulated, the framework
can \emph{prove} statements about them which axiomatic settings can only
assume or must verify model by model. Our main application concerns cubical
\emph{Glue} types: we formulate gluing inside the framework
(\cref{sec:gluing}) and prove that gluing and univalence imply each other in
a precise cycle of implications (\cref{sec:glue-ua}).
This, in turn, allows us to discuss
an approach toward an equivalence between a cubical type theory and book HoTT
(\cref{sec:conservativity}), a version of one of the main open problems of
homotopy type theory.

\paragraph{Agda auto-formalization.}

All results of \cref{sec:framework} to \cref{sec:glue-ua} are formalized using Agda with the \texttt{--two-level} flag~\citep{agdareadthedocs}.
The contents discussed in the conclusions section (\cref{sec:conservativity}) are not, as they discuss theories, translations and provable judgments, which a development working inside one such theory cannot state.
The formalization builds on Uskuplu's 2LTT library~\citep{2LTTAgda,Uskuplu2025} and type checks using
Agda~2.8.0.

The formalization was written by Claude Opus~4.8/5.
I (the paper author) made many requests to change structure and namings, but made only minor manual edits.
I manually checked that the formalized statements are readable and coincide with the paper statements.
The source code consists of approximately 5,700 lines of code (not counting comments and blank lines) and is available as \texttt{.agda} files that can be type-checked locally and as \texttt{.html} files, generated with \texttt{agda-html}, for browsing the formalization without installing Agda.
The code and its HTML rendering are respectively available at
\url{https://github.com/nicolaikraus/extensiontypes-agda} and
\url{https://nicolaikraus.github.io/extensiontypes-agda/html/index.html}.
The file \href{\agdabaseurl}{index.html} acts as an interface between the formal development
and the paper. In particular, all environments that are suitable for formalization in this paper are marked
with a \formalized{} symbol, which is a clickable link to the corresponding
formalized statement in the HTML rendering.

\paragraph{Contributions.}
\begin{itemize}[noitemsep]
	\item We provide a uniform framework that captures the
	various extension types occurring in the literature (\cref{sec:framework}).
	Semantically, the framework shows that the concept of extension types
	is free, in the sense discussed above; briefly, if we have a model of a type theory, we can automatically derive a model of the same theory with extension types.
	\item We demonstrate that this framework makes the rules of \citet{RS2017}
	automatic: the rules hold strictly and the axioms become theorems.
	\item We use the framework to study CCHM-style gluing and prove that it
	is, in a specific formulation, equivalent to univalence; this in particular
	reverses the usual direction that shows that gluing implies univalence
	(\cref{sec:glue-ua}).
	\item One of the main open problems of homotopy type theory is the connection between
	``book HoTT'', the original version of homotopy type theory~\citep{HoTTBook},
	and cubical type theories.
	Our equivalence between gluing and univalence in this 2LTT framework
	allows us to suggest an approach 
	toward this open problem (\cref{sec:conservativity}).
	\item In the existing literature, formalizations in proof assistants usually concern internal results (such as mathematical theorems) \emph{or} meta-theoretic properties (such as normalization of a system).
	The current paper combines both and, building on Uskuplu's 2LTT library \citep{Uskuplu2025}, we demonstrate that the two kinds can be combined in a 2LTT setting.
	All results before the conclusions section are auto-formalized in Agda, cf.\ the above paragraph.
\end{itemize}

\paragraph{Related work.}
\emph{Occurrences of extension types.}
The name and rule package are due to Riehl and Shulman \citep[\S2.2]{RS2017},
who credit unpublished work of Lumsdaine and Shulman; their extension types
are implemented in Rzk \citep{rzk,KRW2024}, and the corresponding lemmas can be found in the formalization~\cite{sHoTT}.
Strictly substitution-stable semantics is given by Weinberger \citep{Wei2022}.
Simplicial extension types drive
the ``HoTT-internal'' synthetic $(\infty,1)$-category theory \citep{WB2023,GWB2024,GWB2025,GWB2026,Toth2025};
these results live in standard homotopy type theory, and the extension types are internal definitions without strict boundary conditions. This can be understood as the degenerated version of our setting where the outer level is HoTT as well, and the conversion map between the levels is the identity.

On the
cubical side, partial elements and their extensions run from the first
cubical model \citep{BCH2014} through Glue types \citep{CoquandTalk,CCHM2018}
and Cartesian cubical extension types \citep{angiuli2019computational} to the
systems Cubical Agda \citep{CubicalAgda}, the two-level RedPRL \citep{RedPRL2018},
redtt \citep{redtt}, cooltt \citep{cooltt}, Aya \citep{aya}, and
Arend \citep{arend}. The proposition-indexed case appears as extent
types \citep{sterling2022first,sterling2021logical}, in controlled unfolding
\citep{gratzer2022controlling}, and in Zhang's account of non-cubical
applications \citep{zhang2023three}.

The FaceTT program of Sch\"onlank, Nuyts, and Devriese
\citep{facett_types_talk}, at the time of writing available only as a talk
abstract, aims to abstract the face-formula layer that these systems share, in
the direction opposite to ours. The present paper shows that in a two-level
setting no dedicated syntax for extension types is needed, and that the rules and
axioms hold automatically once a class of cofibrations is fixed, but it says
nothing about the computational behavior of the resulting system. FaceTT instead
keeps the interval, the faces, and Glue primitive, and parametrizes the syntax by
a \emph{shape theory}, with the goal of obtaining presheaf soundness and
computational properties such as reducibility of face formulas uniformly in that
parameter; this is the syntactic counterpart of the parametrization over a
sub-universe $\Cofib$ discussed above. 

\emph{Study of the tools we use.}
Two-level type theory is essentially Voevodsky's HTS \citep{HTS} plus Capriotti's observation that, when formulated carefully, the theory is conservative over (and thus a good tool to study) HoTT~\citep{Capriotti2017}.
It was introduced in \citet{AltenkirchCapriottiKraus2016}
and developed in detail in \citet{2LTT} as well as the cited thesis by Capriotti.
Our formalization relies on Agda's \texttt{--two-level}
mode \citep{agdareadthedocs} and on Uskuplu's library and results
\citep{2LTTAgda,Uskuplu2025}. Using 2LTT to study conservativity questions was
suggested in \citet{KdJ_representing}.

\emph{Building cubical models, gluing, univalence.}
Closest to our work is Orton--Pitts' internal axiomatization of cofibrations
and filling, used to build cubical models \citep{OrtonPitts2018}: roughly,
our framework is Orton--Pitts over a \emph{fibrant} inner layer, with the
conservativity route back to HoTT; see also
\citet{licata_et_al:LIPIcs.FSCD.2018.22} and \citet{ABCFHL2021}. The Glue-to-univalence
direction belongs to the equivalence-extension line \citep{Sattler2017};
simplicial univalence is proved by representability and extension arguments
\citep{KapulkinLumsdaine2021,GambinoHenry2022}; univalence is decomposed into
simpler principles in \citet{OrtonPitts2017Decomposing}. Shulman's
inverse-diagram gluing \emph{preserves} univalent universes
\citep{Shulman2015InverseDiagrams}, a semantic result rather than an internal type
former; Nuyts--Devriese \citep{NuytsDevriese2020Transpension} compare internal presheaf operators (Glue, Weld,
strictness).

\emph{Conservativity.}
For 2LTT, conservativity is applied to staged compilation by Kov\'acs
\citep{andras-staging} and proved in Bocquet's thesis
\citep[\S6.6]{BocquetThesis}.
Representing object theories
inside a framework follows the logical-framework tradition \citep{HHP1993}.

\section{Two-level type theory as a framework for extension types}\label{sec:framework}

Semantically, a \emph{two-level type theory} in the sense of the cited formulations \citep{HTS,2LTT,AltenkirchCapriottiKraus2016,Capriotti2017}
can be described as \emph{two} CwF structures (the inner and outer model) on a shared category of contexts, together with a
conversion morphism from the inner model to the outer model.
In this section, we give some background and prove 2LTT-internal strictification statements.

\subsection{Background: two-level type theory}
\label{sec:2ltt-recall}

This subsection recalls the two-level type theory (2LTT) of
\citet{2LTT}, which is inspired by Voevodsky's HTS \citep{HTS} and the observation that a careful formulation thereof is conservative over HoTT~\citep{Capriotti2017}.
2LTT is a type theory with two \emph{levels}, each a Martin-L\"of type
theory with its own type formers and universes.
The concrete instantiation we use has the following two levels:
\begin{itemize}[noitemsep]
	\item The \emph{inner} (also referred to as \emph{fibrant}) level is the type theory we are \emph{actually} interested in; in principle, this is simply ``standard'' HoTT.
	For the later comparison with cubical type theories, it is helpful to not explicitly require certain components of HoTT; we assume that this inner theory has intensional identity types $a=b$,
	$\Pi$, $\Sigma$, $+$, $\one$, $\zero$, $\mathbb{N}$, and
	hierarchy of universes $\UU_0:\UU_1:\cdots$ (and, of course, strengthening the theory with additional components such as higher inductive types would not invalidate any results).
	We assume inner function
	extensionality throughout and add univalence explicitly where used.
	This is where HoTT-style mathematics happens.
	\item The \emph{outer} (also referred to as \emph{strict} or \emph{exo}) level has the
	same formers and its own universes $\UUs_j$, but its identity type,
	the \emph{strict equality} $a\seq b$, satisfies UIP and function
	extensionality. 
	It behaves like the equality of a set-level metatheory:
	proof-irrelevant, but still an ordinary type former with the usual
	eliminator. On paper, it would be harmless to assume that it satisfies equality reflection; we do not do this (in order to match the formalization), but we sometimes treat coercions as silent.
	This level can be understood as an ``internalized partial meta-theory'' of the inner level, and it is where one
	reasons \emph{about} inner types, internally.
\end{itemize}
A \emph{conversion} $c$ maps inner types and terms to outer ones. It is assumed to preserve context extension which, as a consequence, means that it preserves $\Pi$ and $\Sigma$ (and the Unit type).
We assume that this preservation is on the nose, as in axioms (T1)--(T2) of
\citet{2LTT}, to avoid coercions.
We keep $c$ implicit and mix the levels
freely. Every inner type is thus in particular an outer type.
Instead of talking about converted inner types, it is convenient to use
\emph{fibrant presentations} (because having such a representation is closed under strict isomorphism; cf.~\cref{rem:no-T3}). A fibrancy witness for an outer type $B$ consists
of an inner type $B^\circ$ and a strict isomorphism
$\rho_B:B\siso c(B^\circ)$; we call $B^\circ$ its \emph{fibrant match}.
Strict constructions and strict equations are read
on $B$, while inner constructions, identity types, equivalences, and homotopy
properties are read on $B^\circ$, with $\rho_B$ inserted when passing between
the two. We normally suppress the matches and the maps $\rho_B$ from notation.
We annotate outer constructions with ${}^s$ and leave fibrant constructions unannotated.%
\footnote{The 2LTT literature is incoherent on the notation: In the original paper \citep{2LTT}, the inner/fibrant fragment is annotated; and in \citet{ahrens2025univalence} and other work, the outer level is annotated with ${}^e$ (for \emph{exo}).}
In particular, a fibrant
type $A$ carries \emph{two} notions of equality: the inner $a=b$ and the finer
strict $a\seq b$. The rules let us derive a coercion $\iota:(a\seq b)\to(a=b)$, but no map
in the other direction; more generally, the inner identity type eliminates
only into fibrant types, whereas $\seq$ eliminates into everything. This
is what makes the theory useful: strict data can always be
read as weak one, and some of our results establish which weak data can be strictified.

Outer types such as
$\mathbb{N}^{\mathrm{s}}$, or the strict intervals and shapes of cubical
and simplicial type theories, are not assumed fibrant: they carry no
inner identity type and no transport, and inner constructions cannot
eliminate into them. We sometimes call outer types \emph{shapes} when they play the role
of domains of extension types; the literature also knows them as \emph{pretypes} or \emph{pseudotypes}.
The interaction of the two levels is mostly captured by 
three notions from \citet[\S3]{2LTT}:
First, a \emph{fibration} is a map all of whose strict fibers are fibrant
\citep[Def.~3.7, Lem.~3.9]{2LTT}; up to strict isomorphism, fibrations are exactly the projections $\Sigma(x:X).Y(x) \to X$ of fibrant families $Y$.
A fibration is \emph{trivial} if all its strict fibers are contractible.
\begin{wrapfigure}{r}{0.21\textwidth}
	\centering
	\vspace{-1.0\baselineskip}
	$\begin{tikzcd}[column sep=small, row sep=normal]
		A  \ar[r,"l"] \ar[d,hookrightarrow,"i"']
		& \Sigma_X Y \ar[d,->>] \\
		B \ar[r, "k"'] & X
	\end{tikzcd}$
	\vspace{-1.1\baselineskip}
\end{wrapfigure}
Second, a \emph{cofibration}, written $i : A \hookrightarrow B$, is a map the Leibniz exponential of which preserves fibrations and trivial fibrations
\citep[Def.~3.13]{2LTT}; spelled out (cf.~\citealp[Rem.~3.14(i)]{2LTT}) this means that, in any strictly commuting square as displayed on the right,
the type of strict diagonal fillers is fibrant, and contractible if the right vertical map is a trivial fibration.
Finally, a type $B$ is \emph{cofibrant} if
$\szero \to B$ is a cofibration, where $\szero$ is the strict
empty type; equivalently, if dependent products
over $B$ preserve (trivial) fibrancy \citep[Cor.~3.20]{2LTT}. Fibrant types
and finite types are cofibrant \citep[Lem.~3.25]{2LTT}.
A simple example of a cofibration is the inclusion $C \to C +^{\mathrm{s}} D$ into a strict binary sum with $D$ cofibrant;
for further examples of cofibrations, we refer to \citet[\S~3.4]{2LTT}. Strict isomorphism (maps in both directions so that both roundtrips are strictly equal to the identity) is denoted by $\siso$.
For a function $f : A \to B$ between fibrant types and $b : B$, we denote by $\sfib f (b) \defeq \Sigma (a:A). fa \seq b$ the strict fiber and by $\hfib f (b) \defeq \Sigma (a:A). fa = b$ the usual (homotopy) fiber; at this point, we deviate from the convention that constructions on the inner level are not annotated by placing the marker $\textup{h}$, but the reason for this is that it distinguishes the notation from the notation for path reversals, which is also denoted by $p^{-1}$.

Several (related) facts make 2LTT a safe setting for work in homotopy type theory. First, 2LTT is modeled essentially wherever HoTT is: the
simplicial model of \citet{KapulkinLumsdaine2021} extends to a model of 2LTT
with all simplicial sets and their strict equality as the outer level
\citep[\S2.5]{2LTT}, and similarly for other presheaf models;
as the introduction describes, 2LTT is essentially free in terms of semantics.
Moreover, 2LTT with (T1)--(T2) is conservative over its inner
level~\citep[Prop.~2.19]{2LTT}; see also
\citet{andras-staging,BocquetThesis} as well as \cref{rem:no-T3} below. This should be interpreted as saying that
the strict level makes more things expressible but does not add new inner
theorems.

Finally, 2LTT is implemented: Agda's \texttt{--two-level}
flag \citep{agdareadthedocs} provides both levels natively, which is what
the formalization accompanying this paper uses, building
on \citet{2LTTAgda,Uskuplu2025}. The extension types of this paper live
exactly at the interface of the two levels: they constrain inner data by
strict equations.

\begin{remark}[Avoiding the repleteness axiom T3 of 2LTT]\label{rem:no-T3}
	Recall that \cite{2LTT} suggests axiom (T3), expressing \emph{repleteness} of the inner fragment:
	any outer type that is strictly isomorphic to a converted inner type is itself a converted inner type.
	This axiom offers convenience and is justified by canonical models such as simplicial sets, but is not covered by the conservativity proofs in the literature cited above.
	While we expect that axiom (T3) can be justified by a general elaboration construction, which essentially consists of \enquote{unfolding} the fibrancy condition everywhere,
	such a result has not yet been formally established.
	Therefore, in the current paper, we do \emph{not} assume this axiom, and instead make use of the more explicit fibrant representations explained above (which itself can be understood as an instance of this elaboration argument).
	Fibrant representations offer the same level of convenience and the literature's conservativity arguments stay directly applicable. 
	All results in this paper that refer to fibrancy can and should be understood under this reading.
	
	The accompanying Agda development implements this elaboration literally. Its
	prelude imports the non-cumulative \texttt{2LTT\_C} interface, whose
	\texttt{isFibrant} record consists of \texttt{fibrant-match} and
	\texttt{fibrant-witness}, and this interface has no (T3) postulate. The modules
	\texttt{CofibFibration}, \texttt{RS48Ext}, and \texttt{GlueDataFibrant}
	respectively construct the extension-type match, formulate extension
	extensionality on that match, and formulate contractibility of Glue data on
	its match. Thus the formalization checks the representation maps that the
	paper-level convention suppresses.
\end{remark}

\subsection{Strictifications}

In type theory, a \emph{strictification} is usually a construction that turns \emph{some} relationship into a strict one: for example, it replaces structure with equivalent structure so that an equality becomes refl; see, for example, \citet{rice_Cayley} and \citet{kaposi2025type}.
The constructions in these examples are internal, but the observation that something is ``strictified'' is meta-theoretic.
What we do here is essentially exactly this; our 2LTT setting allows us to formulate some general principles.
An easy observation that nevertheless motivates the development is the following: 

\begin{lemma}[\flink{Lemma-2-2-strict}]\label{lem:strict-vs-homotopy-fibre}
	If $f : Y \twoheadrightarrow X$ is a fibration and $X$ is fibrant, then the strict fiber $\sfib f (x) \equiv \Sigma(y:Y).f(y)\seq x$ and the (homotopy) fiber $\hfib f (x) \equiv \Sigma(y:Y).f(y)=x$ over any $x : X$ are equivalent, and the equivalence is the canonical map $(y,q) \mapsto (y, \iota q)$.
\end{lemma}
Note that $X$ needs to be fibrant in order for the (homotopy) fiber to be well-defined.
\begin{proof}
	This is essentially Lemma~3.9 of \citet{2LTT}. Note that $Y$ is strictly isomorphic to $\Sigma(x:X). \sfib{f}(x)$. With this representation, one can calculate that the strict and homotopy fiber over $x$ are $\Sigma(x':X).(\sfib{f}(x')) \times (x' =^s x)$ and $\Sigma(x':X).(\sfib{f}(x')) \times (x' = x)$.
	The first is strictly isomorphic (and thus equivalent) to $\sfib{f}(x)$, while the second is equivalent to it.
	Tracing $(y,q)$ yields that it is mapped to $(y,\iota q)$ as required.
\end{proof}

This implies the following strictification result:
\begin{lemma}[\flink{Lemma-2-3} Strictifying diagonal fillers]\label{lem:correct-hom-fillers}
	Consider a strictly commuting square
	\[
		\begin{tikzcd}
			A  \ar[r,"l"] \ar[d,hookrightarrow,"i"]
			& Y \ar[d,->>,"f"] \\
			B \ar[r, "k"] & X
		\end{tikzcd}
	\]
	in which the left vertical map is a cofibration, the right is a fibration, $A$ is cofibrant, and $X$ is fibrant.
	Write $e_0 : f \circ l \seq k \circ i$ for the strict commutativity witness and $e' \colonequiv \iota(e_0) : f \circ l = k \circ i$ for its coercion.
	The type of strict fillers
	\[
		D \;\colonequiv\; \Sigma(h : B \to Y).\;(h \circ i \seq l) \times (f \circ h \seq k)
	\]
	and the type of homotopy fillers
	\[
		D' \;\colonequiv\; \Sigma(h : B \to Y).\;\Sigma(p_1 : h \circ i = l).\;\Sigma(p_2 : f \circ h = k).\;
		\mathrm{ap}_{f \circ -}(p_1)^{-1} \cdot \mathrm{ap}_{- \circ i}(p_2) = e'
	\]
	are both fibrant, and the canonical map
	\[
		D \;\longrightarrow\; D', \qquad
		(h,\,p_1^s,\,p_2^s) \;\mapsto\; (h,\,\iota\,p_1^s,\,\iota\,p_2^s,\,c),
	\]
	is an equivalence, where the coherence cell $c$ is constructed by strict-equality induction.
\end{lemma}
The last component of $D'$ says that the two triangle paths compose to the coerced square path (which comes from a strict equality, so it is ``essentially'' reflexivity).
As in the proof of \cref{lem:strict-vs-homotopy-fibre}, $f$ is strictly isomorphic over $X$ to the display map of the family $x\mapsto\sfib{f}(x)$ of its strict fibers, and $D$, $D'$ and the canonical map transport along that isomorphism; either presentation may be used.
\begin{proof}
	Since $f$ is a fibration and $X$ is fibrant, $Y$ is fibrant; since $A$ is cofibrant and $i$ is a cofibration, $B$ is cofibrant.
	Hence $(A \to Y)$, $(A \to X)$, $(B \to Y)$ and $(B \to X)$ are fibrant, and so are the homotopy pullback $P^h \colonequiv (A \to Y) \times^h_{A \to X} (B \to X)$ and $D'$, both iterated $\Sigma$-types of fibrant types.
	Post-composition $f_* \colonequiv (f \circ {-}) : (A \to Y) \to (A \to X)$ is a fibration since $A$ is cofibrant, and the strict pullback satisfies $P^s \colonequiv (A \to Y) \times^s_{A \to X} (B \to X) \siso \Sigma(m : B \to X).\, \sfib{f_*}(m \circ i)$, so $P^s$ is fibrant \citep[Lem.~3.9]{2LTT}.
	Write $s_0 \colonequiv (l,k,e_0) : P^s$.
	The Leibniz exponential $E \colonequiv \widehat{\mathsf{exp}}(f,i) : (B \to Y) \to P^s$, $h \mapsto (h \circ i,\, f \circ h,\, \mathsf{refl}^s)$, is a fibration since $i$ is a cofibration, and a strict equality $E(h) \seq s_0$ consists of the two triangle equalities together with an equality of square witnesses, unique by outer UIP; so there is a canonical strict isomorphism $\alpha : D \siso \sfib{E}(s_0)$, and $D$ is fibrant.

	Let $\phi : P^s \to P^h$ coerce the third component, $(g,m,e^s) \mapsto (g,m,\iota(e^s))$.
	Fiberwise over $m : B \to X$ it is the map of \cref{lem:strict-vs-homotopy-fibre} for $f_*$ at $m \circ i$, so $\phi$ is an equivalence by totalization \citep[Thm.~4.7.7]{HoTTBook}, and $\phi\,s_0 = (l,k,e')$.
	An inner equality in $P^h$ from $(\phi \circ E)(h)$ to $\phi\,s_0$ consists of $p_1 : h \circ i = l$ and $p_2 : f \circ h = k$ together with the coherence displayed in $D'$, which gives a canonical equivalence $\beta : \hfib{(\phi \circ E)}(\phi\,s_0) \simeq D'$.
	The map of the statement factors as
	\[
		D \;\xrightarrow{\ \alpha\ }\; \sfib{E}(s_0)
		  \;\xrightarrow{\ \kappa\ }\; \hfib{E}(s_0)
		  \;\longrightarrow\; \hfib{(\phi \circ E)}(\phi\,s_0)
		  \;\xrightarrow{\ \beta\ }\; D',
	\]
	with $\kappa$ the coercion $(h,q) \mapsto (h,\iota\,q)$, an equivalence by \cref{lem:strict-vs-homotopy-fibre} for $E$, and the third map $(h,\tilde q) \mapsto (h,\mathrm{ap}_\phi\,\tilde q)$, an equivalence since $\phi$ is one.
	To identify the composite, generalize $l$ and $k$ and apply strict-equality induction to the two triangle witnesses; outer UIP then reduces $e_0$ to reflexivity, and there the composite sends $(h,\mathsf{refl}^s,\mathsf{refl}^s)$ to $h$ with the two reflexivities and the canonical proof of $\mathsf{refl}^{-1}\cdot\mathsf{refl} = \iota\,\mathsf{refl}^s$. So it is that map, $c$ being its coherence cell.
\end{proof}

In the non-dependent form where $X$ is the unit type and $Y$ any fibrant type, the type of strict fillers simplifies to 
$\Sigma(h : B \to Y).\;(h \circ i \seq l)$, and the type of homotopy fillers to $\Sigma(h : B \to Y).\;(h \circ i = l)$.
We give the dependent version of this; note the additional cofibrancy condition on the domain of the cofibration:

\begin{lemma}[\flink{Lemma-2-4} Dependent version]\label{lem:weak=strict-dep}
	Let $i : A \hookrightarrow B$ be a cofibration with $A$ cofibrant, let
	$Y : B \to \UU$ be a fibrant family, and let $l : \prod_{a : A} Y(i a)$
	be a partial section. 
	Then, the type of homotopy extensions
	\begin{equation}
		\Sigma\bigl(h : \prod_{b : B} Y(b)\bigr). (h \circ i = l)
	\end{equation}
	and the type of strict extensions
	\begin{equation}
		\Sigma\bigl(h : \prod_{b : B} Y(b)\bigr). (h \circ i =^s l)
	\end{equation}
	are both fibrant, and equivalent to each other.
\end{lemma}
\begin{proof}
	By \citet[Lem.~3.18]{2LTT} the restriction map $\prod_{b : B} Y(b) \to
	\prod_{a : A} Y(i a)$ is a fibration, since $i$ is a cofibration and $Y$
	is fibrant; its codomain is fibrant because $A$ is cofibrant
	\citep[Cor.~3.20]{2LTT}. For a given $l$, the homotopy and the strict extensions
	are the homotopy and the strict fiber of this map over $l$, hence equivalent by
	\cref{lem:strict-vs-homotopy-fibre}. Both are fibrant: the strict fiber as a
	fiber of a fibration, the homotopy fiber as a homotopy fiber between fibrant
	types.
\end{proof}

\begin{remark}
	\Cref{lem:weak=strict-dep} is \emph{not} simply the special case of \Cref{lem:correct-hom-fillers} where the bottom horizontal map is the identity and the bottom triangle commutes strictly.
	Counterexample: $A \equiv B \equiv \one$, $i \equiv \id$, $Y \equiv X$,
	$f \equiv \id$, $k \equiv l$. The type of strict fillers is contractible, whereas
	the naive type is $\Sigma(y:Y).\,(y = l) \times (y \seq l) \siso \Omega(Y,l)$,
	the loop space, which may be non-trivial. The same example (with both conditions weak) shows that in
	\cref{lem:correct-hom-fillers} the two inner equalities $p_1,p_2$ cannot be taken
	independently: the coherence cell there is essential.
	The correct version of phrasing \Cref{lem:weak=strict-dep} as a fibered statement is more subtle; it is omitted here, but it is formalized (module \texttt{HomFillers} of the supplementary formalization).
\end{remark}

\begin{longversion}
\begin{lemma*}[Fibered version]\label{lem:weak=strict-fib}
	Consider a strictly commuting square
	\begin{equation}
		\begin{tikzcd}
			A \ar[r,"l"] \ar[d,hookrightarrow,"i"'] & Y \ar[d,->>,"f"] \\
			B \ar[r,"k"'] & X
		\end{tikzcd}
	\end{equation}
	in which $i$ is a cofibration, $A$ is cofibrant, and $f$ is a fibration, with
	\emph{no} fibrancy assumption on $X$ or $Y$; write
	$e_0 : f \circ l \seq k \circ i$ for the strict commutativity witness.
	Let $Z(b) \colonequiv \sfib f (k\,b) = \Sigma(y:Y).\,(f(y)\seq k(b))$
	be the strict fiber of $f$ over $k(b)$; this is a fibrant family $Z : B \to \UU$
	\citep[Lem.~3.9]{2LTT}, and $e_0$ turns $l$ into a partial section
	$l' : \prod_{a:A} Z(i\,a)$ via $l'(a) \colonequiv (l\,a,\,e_0\,a)$.
	Then the type of strict fillers
	\[
		D \;\colonequiv\; \Sigma(h : B \to Y).\;(h \circ i \seq l) \times (f \circ h \seq k)
	\]
	is strictly isomorphic to the type of strict extensions of $l'$ in the family
	$Z$, and hence, by \cref{lem:weak=strict-dep}, fibrant and equivalent to the
	type of homotopy extensions
	\[
		D_{\mathrm{fib}} \;\colonequiv\; \Sigma\bigl(h' : \textstyle\prod_{b:B} Z(b)\bigr).\;(h' \circ i = l').
	\]
\end{lemma*}

\begin{proof}
	Using strict function extensionality of the outer level, a map $h : B \to Y$
	together with a witness of $f \circ h \seq k$ is the same as a section
	$h' : \prod_{b:B} Z(b)$, via $h'(b) \colonequiv (h\,b,\,\text{--})$ and back via
	$h \colonequiv \pi_1 \circ h'$. The boundary conditions correspond: a strict
	equality of pairs whose second components inhabit a strict proposition is
	exactly a strict equality of the first components, so $h \circ i \seq l$ iff
	$h' \circ i \seq l'$. This identifies $D$ with
	$\Sigma(h' : \prod_B Z).\,(h' \circ i \seq l')$, and
	\cref{lem:weak=strict-dep}, applied to the cofibration $i$ (with $A$
	cofibrant) and the fibrant family $Z$, gives fibrancy and the equivalence with
	$D_{\mathrm{fib}}$.
\end{proof}

\begin{remark}[The weak condition must live in the fibers]\label{rem:fibred-warning}
	The weak condition in $D_{\mathrm{fib}}$ is an inner equality in $\prod_{a:A} Z(i\,a)$:
	pointwise inner equalities \emph{in the strict fibers} of $f$, coherent with the strict
	witnesses. It cannot be replaced by the naive pair
	$(h \circ i = l) \times (f \circ h \seq k)$. For non-fibrant $Y$ the first
	component is not even well-formed, being the inner identity type of the
	non-fibrant $A \to Y$; and for fibrant $Y$ the naive type is genuinely larger.
	Counterexample: $A = B = \one$, $i = \id$ (a cofibration), $Y = X$ fibrant,
	$f = \id$, $k \seq l$. The strict fillers form a contractible type, whereas
	the naive type is $\Sigma(y:Y).\,(y = l) \times (y \seq l) \siso \Omega(Y,l)$,
	the loop space. The same example (with both conditions weak) shows that in
	\cref{lem:correct-hom-fillers} the two inner equalities $p_1,p_2$ cannot be taken
	independently: the coherence cell there is essential.
\end{remark}
\end{longversion}

\section{The Riehl--Shulman extension-type calculus}
\label{sec:rs-calculus}

Riehl and Shulman's simplicial type theory \citep{RS2017} is a type theory for
developing the theory of $\infty$-categories synthetically: beneath homotopy
type theory sits a layer of \emph{shapes}, carved out of directed cubes by a
tope logic, and the one new type former is the \emph{extension type}, whose
elements are sections of a family over a shape that restrict, strictly, to
a partial section on a subshape.
In the framework of this paper,
both ingredients are automatically available: the shapes and shape inclusions as
cofibrations,%
\footnote{The claim is that the shape calculus with its rules is automatic; the existence of specific shape inclusions, such as horn inclusions, is an additional assumption.}
and the extension type is definable rather than postulated.


The following shows the rules of extension types as introduced and presented in \citet[Fig.~4]{RS2017}.
We simplify the presentation and omit assumptions which state that something is a valid shape, a valid shape inclusion, or similar; instead, we assume for all rules that $i:\Phi\hookrightarrow\Psi$ is a shape inclusion, with $\Phi$ and $\Psi$ shapes.
Later rules omit assumptions that can be inferred from earlier rules, and the symbol $\seq$ is used to emphasize that all rules are strict.

\vspace*{.5cm}

\begin{mdframed}
	\begin{mathpar}
	\inferrule*[right=form]{A:\Psi\to\UU \\ a:\textstyle\prod_{\phi:\Phi}A(i\phi)}{\Ext{i}{\Pi_\Psi A}{a}\ \mathsf{type}}
	\and
	\inferrule*[right=intro]{b:\textstyle\prod_{\psi:\Psi}A(\psi) \\ b\circ i\seq a}{\lambda\psi.\,b(\psi):\Ext{i}{\Pi_\Psi A}{a}}
	\and
	\inferrule*[right=eval]{f:\Ext{i}{\Pi_\Psi A}{a} \\ \psi:\Psi}{f(\psi):A(\psi)}
	\and
	\inferrule*[right=bdry]{f:\Ext{i}{\Pi_\Psi A}{a} \\ \phi:\Phi}{f(i\phi)\seq a(\phi)}
	\and
	\inferrule*[right=$\beta$]{b:\textstyle\prod_{\psi:\Psi}A(\psi) \\ b\circ i\seq a \\ \psi:\Psi}{(\lambda\psi'.\,b(\psi'))(\psi)\seq b(\psi)}
	\and
	\inferrule*[right=$\eta$]{f:\Ext{i}{\Pi_\Psi A}{a}}{f\seq\lambda\psi.\,f(\psi)}
\end{mathpar}
\captionof{figure}{The rules of extension types as presented in \citet[Fig.~4]{RS2017}.}
\end{mdframed}

\vspace*{.5cm}

In our framework, shapes become cofibrant types, and shape inclusions are simply cofibrations.%
\footnote{Regarding the figure, note that Riehl--Shulman decorate the bar with the subshape $\Phi$, with the inclusion being implicit; here, we consider it helpful to instead decorate the bar with the cofibration/shape inclusion.}
For the rest of this section, we work with the following convention: $i : \Phi \hookrightarrow \Psi$ is a cofibration,
$A : \Psi \to \UU$ a fibrant family, $a : \Pi_{\phi:\Phi} A(i \phi)$ a partial section, and we
write $f \circ i \colonequiv \lambda\phi.\,f(i\phi)$ for restriction along $i$.
Cofibrancy of $\Phi$ or $\Psi$ is not assumed (unless a statement says so).
The formation rule is interpreted by the following definition:

\begin{definition}[\flink{Definition-3-1} Extension type]\label{def:ext-formal}
	Given a cofibration $i : \Phi \hookrightarrow \Psi$, a type family $A : \Psi \to \UU$, and a partial section $a : \Pi_{\phi:\Phi} A(i \phi)$, the \emph{extension type} is defined as
	\begin{equation}
		\Ext{i}{\Pi_\Psi A}{a}\;\colonequiv\;\Sigma\bigl(f:\textstyle\prod_{\psi:\Psi}A(\psi)\bigr).\,(f\circ i\seq a).
	\end{equation}
	This is the strict fiber of restriction along $i$ over the partial section, $\sfib{(i^*)}(a)$.
	Since $i$ is a cofibration, the restriction map $i^*$ is a fibration, so this strict fiber comes equipped with a fibrant match \citep[Lem.~3.9, 3.18]{2LTT}.
	Following the convention of omitting conversions (cf.~\cref{sec:2ltt-recall}), the same notation denotes this match in inner constructions and the displayed realization in strict constructions.
\end{definition}

The special cases (2.3) and (2.4) of \citet{RS2017} are immediate strict isomorphisms: for
$\Phi \equiv \szero$ the boundary component is a contractible strict
proposition, so $\Ext{i}{\Pi_\Psi A}{\mathsf{rec}_\bot}\siso\prod_\Psi A$, and a
constant family gives the plain function type $\Psi\to X$.

The remainder of this section derives the extension-type calculus of
\citet[\S2.2, \S4]{RS2017} from \cref{def:ext-formal}, statement by statement:
the rules of their Figure~4 hold strictly (\cref{lem:rs-rules}); the structural equivalences from their
statements 4.1--4.5 become strict isomorphisms
(\cref{lem:rs41,lem:rs42,lem:rs43,lem:rs44,lem:rs45}); and their relative function extensionality axiom \citep[Axiom~4.6]{RS2017} becomes a theorem
(\cref{lem:rel-funext-holds}), from which the extensionality properties of their
results 4.8--4.12 follow (\cref{lem:rs48,lem:rs410,lem:rs412}; their
4.11 becomes redundant), as does the homotopy extension property \citep[Prop.~4.10]{RS2017}.

Nothing beyond standard function extensionality is assumed.
The exceptions are the union law \citep[Thm.~4.5]{RS2017}, which needs the cofibration class to be
closed under binary meets; arbitrary cofibrations need not have this property, so the assumption has to be added explicitly (\cref{lem:rs45}).%
\footnote{The simplicial/cubical face formulas do have it: the face lattice is closed under conjunction.}
Two statements also need the outer level to have the pushouts they mention: \cref{lem:rs42} the pushout-product, \cref{lem:rs45} the join $\Phi\vee\Psi$.


\subsection{The rules and the proposition-indexed calculus}

The first lemma makes \cref{def:ext-formal} meaningful:

\begin{lemma}[\flink{Lemma-3-2-intro} Riehl--Shulman rules hold]\label{lem:rs-rules}
	The extension type of \cref{def:ext-formal} validates the rules of
	\citet[Fig.~4]{RS2017}.
\end{lemma}
\begin{proof}
	By \cref{def:ext-formal} an element of $\Ext{i}{\Pi_\Psi A}{a}$ is a pair $(f,p)$
	with $f:\Pi_\Psi A$ and $p: f\circ i\seq a$. Thus, \textsc{intro} sends $b$ (with its unique boundary proof) to
	$(b,\,\text{--})$; \textsc{eval} is $f(\psi)\colonequiv\pi_1(f)(\psi)$;
	\textsc{bdry} is $\pi_2(f)$ read at $\phi$; and $\beta$, $\eta$ are those of the
	strict $\Pi$-types together with the fact that the boundary component is a strict proposition, so both round-trips are strict
	identities.
\end{proof}

\begin{remark}[\flink{Remark-3-3-type} The proposition-indexed calculus]\label{rem:zhang-rules}
	The extension type of \cref{def:ext-formal} also subsumes the
	\emph{proposition-indexed} extension types of Zhang \citep{zhang2023three} and
	of the controlled-unfolding calculus \citep{gratzer2022controlling}. Their
	connective $\{A\mid\varphi\rhd u\}$ is the extension type along the map
	$\varphi\to\one$, and its rules are \cref{lem:rs-rules} for that map, so no
	separate verification is required. (The displayed outer realization in \cref{def:ext-formal}
	and the strict rules of \cref{lem:rs-rules} never use the cofibration property
	of $i$, so they apply to $\varphi\to\one$ whether or not this map is a
	cofibration; the fibrant match of \cref{def:ext-formal} does use it.) The only feature that needs its own
	comment is the treatment of several clauses at once, which we reach at the end.

	These calculi are organized around a universe $F$ of proof-irrelevant
	propositions and \emph{restricted judgments} $\Gamma,\varphi\vdash J$, meaning
	that $J$ is checked under the assumption $\varphi:F$. The connective
	$\{A\mid\varphi\rhd u\}$ is a wrapper with rules \citep[Fig.~2.1]{zhang2023three}:
	\begin{mathpar}
		\inferrule*[right=form]{\Gamma\vdash A\ \mathsf{type} \\ \Gamma\vdash\varphi:F \\ \Gamma,\varphi\vdash u:A}{\Gamma\vdash\{A\mid\varphi\rhd u\}\ \mathsf{type}}
		\and
		\inferrule*[right=intro]{\Gamma\vdash v:A \\ \Gamma,\varphi\vdash u\equiv v:A}{\Gamma\vdash\mathsf{inS}(v):\{A\mid\varphi\rhd u\}}
		\and
		\inferrule*[right=elim]{\Gamma\vdash v:\{A\mid\varphi\rhd u\}}{\Gamma\vdash\mathsf{outS}(v):A}
		\and
		\inferrule*[right=bdry]{\Gamma\vdash v:\{A\mid\varphi\rhd u\}}{\Gamma,\varphi\vdash\mathsf{outS}(v)\equiv u:A}
		\and
		\inferrule*[right=$\beta$]{\Gamma\vdash v:A}{\Gamma\vdash\mathsf{outS}(\mathsf{inS}(v))\equiv v:A}
	\end{mathpar}
	The $\eta$-rule is given by \citet[Ex.~2.6]{zhang2023three}. We read all of this off
	\cref{def:ext-formal,lem:rs-rules} at $i=(\varphi\to\one)$, the strict
	injection of a strict proposition, under the dictionary: $F$ is a universe of
	strict propositions; a restricted judgment $\Gamma,\varphi\vdash J$ is
	$\Gamma,p:\varphi\vdash J$; a restricted judgmental equality
	$\Gamma,\varphi\vdash u\equiv v:A$ is an element of $\prod_{p:\varphi}u\seq v$,
	itself a strict proposition; and the connective is the strict fiber
	\[
		\{A\mid\varphi\rhd u\}\;\rightsquigarrow\;
		\Ext{\varphi\to\one}{\textstyle\prod_{\one}A}{u}
		\;\;\equiv\;\;\Sigma(x:A).\,\textstyle\prod_{p:\varphi}x\seq u ,
	\]
	writing $A$, $u$ also for the constant family over $\one$ and the partial
	element $\lambda p.u$. Rule by rule, Zhang's \textsc{form} is the formation of
	this extension type (\cref{def:ext-formal}), and \textsc{intro}, \textsc{elim},
	\textsc{bdry}, $\beta$, $\eta$ are \textsc{intro}, \textsc{eval}, \textsc{bdry},
	$\beta$, $\eta$ of \cref{lem:rs-rules}: concretely $\mathsf{inS}$ is pairing,
	$\mathsf{outS}=\pi_1$, \textsc{bdry} is $\pi_2$, and $\beta,\eta$ are those of
	the strict $\Sigma$-type, its boundary component $\prod_{p:\varphi}x\seq u$
	being a strict proposition. These are all of Zhang's rules for a single clause.

	The one feature that does not fit this single-clause reading is that Zhang
	attaches several clauses $\varphi_1\rhd u_1,\dots,\varphi_n\rhd u_n$ at once,
	with \textsc{intro} asking for agreement with every $u_k$ and \emph{no}
	compatibility condition between them. This is the same construction along the
	copairing $[\varphi_1,\dots,\varphi_n]:\varphi_1+\dots+\varphi_n\to\one$. The
	domain $\varphi_1+\dots+\varphi_n$ need not be a strict proposition, but the
	boundary type $\prod_{p}\bigl(x\seq[u_1,\dots,u_n](p)\bigr)$, with $p$ ranging
	over that sum, still is one; so any inhabitant $x$ satisfies $x\seq u_j$ and
	$x\seq u_k$ wherever $\varphi_j$ and $\varphi_k$ both hold. Agreement of the
	$u_k$ on overlaps is therefore a consequence of inhabitation rather than a side
	condition, which is exactly why Zhang's rules omit it.

	The reading is not tied to $\Psi=\one$: fiberwise over the codomain, every
	extension type is a $\prod$ of extension types along maps into $\one$
	\flink{Remark-3-3-fibrewise}. Write
	$\sfib{i}(\psi)\defeq\Sigma(\phi:\Phi).\,(i\phi\seq\psi)$ for the strict
	fiber and $a_\psi(\phi,p)\defeq p_*(a\phi)$ for the boundary datum
	transported into it along $p$. Then
	\[
		\Ext{i}{\textstyle\prod_{\psi:\Psi}A\psi}{a}
		\;\siso\;
		\textstyle\prod_{\psi:\Psi}\{A\psi\mid\sfib{i}(\psi)\rhd a_\psi\},
	\]
	by strict-equality induction in one direction and $\prod$-extensionality in
	the other. The factors are the connective of this remark at a general index,
	as in the several-clauses case above. They are \emph{proposition}-indexed
	exactly when the fibers $\sfib{i}(\psi)$ are strict propositions, that is
	when $i$ is injective, and injectivity is a property of the boundary
	inclusions of particular calculi rather than part of the cofibration
	property of \cref{sec:2ltt-recall}.\footnote{%
		The Leibniz condition does not entail it. For
		$i:\one+^{\mathrm{s}}\one\to\one$ and $f:Y\to X$, the strict fiber of
		$\widehat{\mathsf{exp}}(f,i)$ over $(y_0,y_1,x)$ is $y_0\seq y_1$: if $f$
		is a trivial fibration this is contractible, since $y_0$ and $y_1$ lie in
		the contractible strict fiber $\sfib{f}(x)$ and strict equality types are
		strict propositions, and it is fibrant whenever every type is, that is
		when the inner level is the whole outer one. In that interpretation this
		$i$ is a cofibration and is not injective.} Simplicial subshape inclusions and the
	extents of cubical face formulas do have it, so there the shape-indexed and
	the proposition-indexed presentations differ in where the base sits, in an
	explicit shape or in an ambient $\prod$, rather than in the construction.

	The cubical calculi make that split explicitly.
	\citet[\S3.5]{angiuli2019computational} has line types $(x{:}I)\to A$ and
	restriction types $A[\xi\hookrightarrow N]$ as separate formers and defines
	$\mathsf{Path}_{x.A}(a,a')$ as
	$(x{:}I)\to A[x=0\hookrightarrow a,\,x=1\hookrightarrow a']$; restriction
	types are not Kan on their own, while the combination is Kan once the line
	binds the variables the face formula mentions, so it is the combination
	$(\vec{x}{:}I)\to A[\vec{\xi}\hookrightarrow\vec{N}]$ that is taken as the
	primitive extension type, its restriction part being the connective above.
	What varies between the settings is the face calculus, which boundaries can
	be named at all, rather than the shape of the construction.
\end{remark}

\subsection{The structural equivalences (RS 4.1--4.5)}

The five structural equivalences of \citet[\S4, Thms~4.1--4.5]{RS2017} are all provable;
only the union law asks for more than an arbitrary cofibration.

\begin{lemma}[\flink{Lemma-3-4} RS 4.1: commutation with  $\Pi$]\label{lem:rs41}
	For $X:\UU$, a fibrant $Y:\Psi\to X\to\UU$ and
	$f:\prod_{\phi:\Phi}\prod_{x:X}Y(i\phi,x)$,
	\[
		\Ext{i}{\textstyle\prod_{\psi:\Psi}\prod_{x:X}Y(\psi,x)}{f}\;\siso\;\textstyle\prod_{x:X}\Ext{i}{\textstyle\prod_{\psi:\Psi}Y(\psi,x)}{\lambda\phi.f(\phi,x)}.
	\]
\end{lemma}
\begin{proof}
	Currying is a strict isomorphism $\prod_\psi\prod_x Y\siso\prod_x\prod_\psi Y$,
	$g\mapsto\lambda x.\lambda\psi.g(\psi)(x)$. Under it $g\circ i\seq f$ holds iff
	for every $x$, $(\lambda\psi.g(\psi)(x))\circ i\seq\lambda\phi.f(\phi,x)$ (a
	strict equality of functions is pointwise, its proofs irrelevant), so the
	strict fiber over $f$ on the left is the product over $x$ of the strict fibers
	on the right; the maps are inverse by $\beta/\eta$.
\end{proof}

\begin{lemma}[\flink{Lemma-3-5} RS 4.2: currying / pushout-product]\label{lem:rs42}
	Let $i:\Phi\hookrightarrow\Psi$ and $i':\Phi'\hookrightarrow\Psi'$ be
	cofibrations, so the pushout-product $i\poprod i'$, with domain the pushout
	$\Phi\times\Psi'\cup_{\Phi\times\Phi'}\Psi\times\Phi'$, is a cofibration
	by \citet[Lem.~3.23]{2LTT}, provided this pushout exists in the outer level.
	Let $A$ be a fibrant family over $\Psi\times\Psi'$ and $c$ a partial section
	over that pushout; we write $c$ for both of its components, so
	$c(\psi,-):\prod_{\phi':\Phi'}A(\psi,i'\phi')$ is the boundary datum of the
	inner extension type at $\psi:\Psi$. Since $c$ is total in the second argument
	over $\Phi\times\Psi'$, it induces the outer boundary
	\[
		c_1:\textstyle\prod_{\phi:\Phi}\Ext{i'}{\Pi_{\Psi'}A(i\phi,-)}{c(i\phi,-)},
		\qquad
		c_1(\phi)\colonequiv\bigl(\lambda\psi'.\,c(\phi,\psi'),\;\text{--}\bigr),
	\]
	whose second component is the strict boundary witness supplied by the strict
	agreement of the two components of $c$ over $\Phi\times\Phi'$ (the pushout
	cocone condition); symmetrically
	$c_2:\prod_{\phi':\Phi'}\Ext{i}{\Pi_\Psi A(-,i'\phi')}{c(-,i'\phi')}$.
	Then
	\[
		\Ext{i}{\textstyle\prod_{\psi:\Psi}\Ext{i'}{\Pi_{\Psi'}A(\psi,-)}{c(\psi,-)}}{c_1}\;\siso\;\Ext{i\poprod i'}{\Pi_{\Psi\times\Psi'}A}{c}\;\siso\;\Ext{i'}{\textstyle\prod_{\psi':\Psi'}\Ext{i}{\Pi_\Psi A(-,\psi')}{c(-,\psi')}}{c_2}.
	\]
\end{lemma}
\begin{proof}
	By currying and swapping the two arguments, called \enquote{commutation of
	arguments} by Riehl and Shulman \citep[\S4.1]{RS2017},
	$\prod_{\Psi\times\Psi'}A\siso\prod_\Psi\prod_{\Psi'}A\siso\prod_{\Psi'}\prod_\Psi A$
	as strict isomorphisms (by $\eta$). By the universal property of the pushout, a
	boundary on $\Phi\times\Psi'\cup_{\Phi\times\Phi'}\Psi\times\Phi'$ is a pair of
	boundaries, on $\Phi\times\Psi'$ and on $\Psi\times\Phi'$, agreeing on
	$\Phi\times\Phi'$; under these isomorphisms this is exactly the boundary of the
	iterated extension type. Both isomorphisms follow, the two orders of iteration
	giving the outer terms.
\end{proof}

\begin{lemma}[\flink{Lemma-3-6} RS 4.3: extending into $\Sigma$]\label{lem:rs43}
	For fibrant $X:\Psi\to\UU$ and $Y:(\sum_\Psi X)\to\UU$, with $a:\prod_\Phi X\circ i$
	and $b:\prod_{\phi}Y(i\phi,a(\phi))$,
	\[
		\Ext{i}{\textstyle\prod_{\psi:\Psi}\sum_{x:X(\psi)}Y(\psi,x)}{(a,b)}\;\siso\;\textstyle\sum_{f:\Ext{i}{\Pi_\Psi X}{a}}\Ext{i}{\textstyle\prod_{\psi:\Psi}Y(\psi,f\psi)}{b}.
	\]
\end{lemma}
\begin{proof}
	A section $h:\prod_\psi\sum_{x}Y(\psi,x)$ is, by the projections, a pair
	$(f,g)$ with $f\colonequiv\lambda\psi.\pi_1(h\psi)$ and
	$g\colonequiv\lambda\psi.\pi_2(h\psi)$; this is a strict isomorphism with
	inverse $(f,g)\mapsto\lambda\psi.(f\psi,g\psi)$. The boundary $h\circ i\seq(a,b)$
	holds iff $f\circ i\seq a$ and $g\circ i\seq b$ componentwise, so the strict
	fiber splits as the displayed $\Sigma$.
\end{proof}

\begin{lemma}[\flink{Lemma-3-7} RS 4.4: composite of cofibrations]\label{lem:rs44}
	For cofibrations $i:\Phi\hookrightarrow\Psi$, $j:\Psi\hookrightarrow\Omega$,
	a fibrant family $A:\Omega\to\UU$ and a partial section
	$a:\prod_{\phi:\Phi}A(j(i\phi))$,
	\[
		\Ext{j\circ i}{\Pi_\Omega A}{a}\;\siso\;\textstyle\sum_{f:\Ext{i}{\Pi_\Psi(A\circ j)}{a}}\Ext{j}{\Pi_\Omega A}{f}.
	\]
\end{lemma}
\begin{proof}
	$j\circ i$ is a cofibration (composition closure). Map $h\mapsto(h\circ j,h)$:
	$h\circ j$ satisfies $(h\circ j)\circ i \seq h\circ(ji)\seq a$, so
	$h\circ j\in\Ext{i}{\Pi_\Psi(A\circ j)}{a}$, and $h\circ j\seq h\circ j$ witnesses
	$h\in\Ext{j}{\Pi_\Omega A}{h\circ j}$. The inverse $(f,h)\mapsto h$ uses
	$h\circ(ji) \seq f\circ i\seq a$; round-trips are identities since $f$ is forced to
	be $h\circ j$.
\end{proof}

\begin{lemma}[\flink{Lemma-3-8} RS 4.5: union of cofibrations]\label{lem:rs45}
	Let $\Phi_1$ and $\Phi_2$ be subshapes of a common shape, and suppose the
	cofibration class is closed under binary meets, so that
	$\Phi_1\wedge\Phi_2\hookrightarrow\Phi_1$ is a cofibration and $\Phi_1\vee\Phi_2$ is the
	pushout of $\Phi_1\hookleftarrow\Phi_1\wedge\Phi_2\hookrightarrow\Phi_2$. Then for
	fibrant $A$ over $\Phi_1\vee\Phi_2$ and $a:\prod_{\Phi_2} A|_{\Phi_2}$,
	\[
		\Ext{\Phi_2\hookrightarrow\Phi_1\vee\Phi_2}{\Pi_{\Phi_1\vee\Phi_2}A}{a}\;\siso\;\Ext{\Phi_1\wedge\Phi_2\hookrightarrow\Phi_1}{\Pi_{\Phi_1} A}{a|_{\Phi_1\wedge\Phi_2}}.
	\]
	Without meet-closure the right-hand side need not be an extension type
	along a cofibration: $\Phi_1\wedge\Phi_2\hookrightarrow\Phi_1$, a base change of
	$\Phi_2\hookrightarrow\Phi_1\vee\Phi_2$, need not be a cofibration, as cofibrations are not provably pullback-stable.
\end{lemma}
\begin{proof}
	By the universal property of $\Phi_1\vee\Phi_2$, a section over it is a pair of
	sections over $\Phi_1$ and $\Phi_2$ agreeing on $\Phi_1\wedge\Phi_2$. Fixing the
	$\Phi_2$-component to $a$ leaves a section over $\Phi_1$ agreeing with $a$ on
	$\Phi_1\wedge\Phi_2$. The maps are restriction and the copairing
	$\mathrm{rec}_\vee$, inverse by $\beta/\eta$.
\end{proof}

\subsection{Relative function extensionality and realignment (RS 4.6)}

Relative function extensionality (RS's Axiom 4.6) is \emph{not} an extra
assumption in the framework: it holds automatically,
because our shape inclusions are cofibrations and a cofibration preserves
\emph{trivial} fibrations (\cref{lem:rel-funext-holds}). In our notation:

\begin{definition}[\flink{Definition-3-9} Relative function extensionality, RS Axiom 4.6]\label{def:rel-funext}
	Say $(i,\UU)$ satisfies \emph{relative funext} if, for every fibrant
	$A:\Psi\to\UU$ with each $A(\psi)$ contractible and every
	$a:\prod_\Phi A\circ i$, the type $\Ext{i}{\Pi_\Psi A}{a}$ is contractible.
\end{definition}

\noindent
This is the \emph{contractibility} form of extensionality for extension types.
There are two further forms, an equivalence form and a naive one, and RS take
the contractibility form as their axiom because they can derive the other two
from it while not knowing how to derive it from either of them alone
\citep[\S4.4]{RS2017}. Here all three are theorems: this one is
\cref{lem:rel-funext-holds}, and the other two follow from it in
\cref{lem:rs48}.

\begin{remark}[What relative funext is]\label{rem:relfunext-is-shape-funext}
	For general $i$, \cref{def:rel-funext} is exactly the
	\emph{trivial-fibration} half of the cofibration property of $i$, read
	on the strict fibers of the restriction map; this is why it holds for
	our shape inclusions (\cref{lem:rel-funext-holds}). The instance
	$i = \szero \hookrightarrow\Psi$, for cofibrant $\Psi$, explains the
	\emph{function extensionality} in the name:
	there it says that $\prod_\Psi A$ is contractible whenever each
	$A(\psi)$ is, i.e.\ \emph{function extensionality for the shape-indexed
	product $\prod_\Psi$}, in contractibility form. This is
	\emph{not} a consequence of inner funext (which covers $\Pi$ over inner
	types only), nor even of inner $+$ outer funext (outer funext concerns
	$\seq$, not the inner identity type).
\end{remark}

\begin{lemma}[\flink{Lemma-3-11} RS 4.6: relative function extensionality holds for cofibrations]\label{lem:rel-funext-holds}\label{lem:rs46}
	Every cofibration $i:\Phi\hookrightarrow\Psi$ satisfies relative
	funext; in particular, every shape inclusion does. Explicitly: if each
	$A(\psi)$ is contractible, then $\Ext{i}{\Pi_\Psi A}{a}$ is contractible. No
	cofibrancy of $\Phi$ or $\Psi$ is needed.
\end{lemma}
\begin{proof}
	If each $A(\psi)$ is contractible, then $A$ is a family of \emph{trivially
	fibrant} types. Since $i$ is a cofibration, restriction along $i$
	carries such families to \emph{trivial} fibrations \citep[Lem.~3.18]{2LTT}:
	$i^*:\prod_\Psi A\to\prod_\Phi A\circ i$ is a trivial fibration. Its strict
	fibers are therefore trivially fibrant \citep[Def.~3.7]{2LTT}; in
	particular the fiber over $a$, which is $\Ext{i}{\Pi_\Psi A}{a}$, is
	contractible. This uses exactly the trivial-fibration half of the
	cofibration property (see \cref{sec:2ltt-recall}); the fibration half
	alone does not give this argument. 
\end{proof}

\begin{longversion}
\begin{remark}[Why RS assume it]\label{rem:rs-assumes-extext}
	\Cref{lem:rel-funext-holds} is an instance of the paper's thesis that axioms
	assumed elsewhere are automatic here. RS postulate relative funext (their Axiom
	4.6) because their cofibrations are bare \emph{shape inclusions}, carrying no
	preservation of trivial fibrations; ours are the full 2LTT notion, which does.
	The \emph{fibration} half of cofibrancy alone (\citealp[Cor.~3.20]{2LTT}:
	$\prod_\Psi$ preserves fibrancy) is \emph{not} enough.
\end{remark}
\end{longversion}

Relative funext has a consequence that we will use repeatedly: components of
a structure that are \emph{property-like} in the HoTT sense can be strictly
realigned along any cofibration. In HoTT, since such components are
propositions, they are invisible to \emph{inner equalities} between structures; they
are not invisible to \emph{strict} equality, since inhabitants of
propositions are propositionally, not strictly, unique. The following says
that along a cofibration this deficit can always be repaired, using nothing
beyond the cofibration property of $i$; no cofibrancy of $\Phi$ or $\Psi$
is needed.

\begin{corollary}[\flink{Corollary-3-12} Strict realignment of propositional components]\label{cor:prop-realign}
	Let $i:\Phi\hookrightarrow\Psi$ be a cofibration, $D:\Psi\to\UU$ a
	fibrant family (\emph{data}), and let $P(\psi,d)$ be a fibrant family of
	inner propositions, for $\psi:\Psi$ and $d:D(\psi)$ (\emph{property}).
	Write $S(\psi)\colonequiv\Sigma(d:D\,\psi).\,P(\psi,d)$. Given a section
	$s:\prod_\Psi S$ and a partial section $s_0:\prod_\Phi S\circ i$ whose
	underlying data agree strictly with $s$, i.e.\
	$\pi_1(s(i\phi))\seq\pi_1(s_0\,\phi)$ for all $\phi:\Phi$, the type of
	\emph{realignments}
	\[
		\Sigma\bigl(s':\prod_\Psi S\bigr).\;
		\bigl(\pi_1\circ s'\seq\pi_1\circ s\bigr)\times
		\bigl(s'\circ i\seq s_0\bigr)
	\]
	is contractible. In particular a realignment exists, changing only the
	property components of $s$, and it is unique up to homotopy.
\end{corollary}
\begin{proof}
	Set $P'(\psi)\colonequiv P(\psi,\pi_1(s\,\psi))$. Each $P'(\psi)$ is a
	proposition inhabited by $\pi_2(s\,\psi)$, hence contractible, so $P'$ is
	a fiberwise contractible fibrant family. Transporting $\pi_2(s_0\,\phi)$
	along the assumed strict equality of underlying data yields a partial
	section $l:\prod_\Phi P'\circ i$. Since the strict equality of a
	$\Sigma$-pair decomposes into strict equalities of the components (the
	second transported along the first), and since the constraint
	$\pi_1\circ s'\seq\pi_1\circ s$ identifies $s'$ with a section of $P'$
	(using UIP of the outer level), the type of realignments is strictly
	isomorphic to $\Ext{i}{\Pi_\Psi P'}{l}$, which is contractible by
	\cref{lem:rel-funext-holds}.
\end{proof}

\noindent
Any property-like component qualifies: $\mathsf{isEquiv}$ (used below for
the boundary coherence of Glue structures, \cref{def:glue-structure}),
$\mathsf{isContr}$, being an $n$-type, half-adjointness data, or a
mere-existence witness.

\subsection{Extension extensionality and truncation levels (RS 4.7--4.12)}

With relative function extensionality in hand, the remaining statements of
\citet[\S4]{RS2017} follow by RS's own arguments.

\begin{lemma}[\flink{Lemma-3-13} RS 4.8: extension extensionality]\label{lem:rs48}
	For $f,g:\Ext{i}{\Pi_\Psi A}{a}$, the map
	$\mathsf{happly}:(f=g)\to\Ext{i}{\prod_{\psi}(f\psi=g\psi)}{\lambda\phi.\mathsf{refl}}$
	is an equivalence.
\end{lemma}
\begin{proof}
	Fix $f$; it suffices that the induced map of total spaces over $g$ is an
	equivalence. Its domain $\sum_g(f=g)$ is contractible (a based identity type). Its
	codomain
	$\sum_g\Ext{i}{\prod_{\psi}(f\psi=g\psi)}{\lambda\phi.\mathsf{refl}}$ is,
	by \cref{lem:rs43}, equivalent to
	$\Ext{i}{\prod_{\psi}\sum_{y:A(\psi)}(f\psi=y)}{\lambda\phi.(a\phi,\mathsf{refl})}$;
	each $\sum_{y}(f\psi=y)$ is contractible, so this is contractible by
	\cref{lem:rs46}. A map of contractible types is an equivalence.
\end{proof}

The map of \cref{lem:rs48} is RS's (4.7) and the lemma their Proposition
4.8(i). Its domain is the identity type \emph{of the extension type}, which the
lemma identifies with relative homotopies restricting to $\mathsf{refl}$ on
$\Phi$: that is the \emph{relative} in the name (RS Remark 4.9), the
\emph{function extensionality} being the instance $\Phi=\szero$ that
\cref{rem:relfunext-is-shape-funext} reads off the contractibility form. This
also completes the three forms announced after \cref{def:rel-funext}: the
contractibility form is \cref{def:rel-funext}, a theorem by \cref{lem:rs46};
the equivalence form is \cref{lem:rs48}; and the naive form (RS 4.8(ii): two
extensions that are pointwise equal relative to $\Phi$ are equal) follows by
inverting it. In the Rzk formalization \citep{KRW2024} the three are
\texttt{WeakExtExt}, \texttt{ExtExt} and \texttt{NaiveExtExt}, with
\texttt{ExtExt}$\Leftrightarrow$\texttt{WeakExtExt} observed by T.\ Walde.

\begin{lemma}[\flink{Lemma-3-14} RS 4.10: the homotopy extension property]\label{lem:rs410}\label{cor:RS-equiv-ext}
	Let $i : \Phi \hookrightarrow \Psi$ be a cofibration,
	let $Y : \Psi \to \UU$ be a fibrant family, let
	$b : \prod_{\psi:\Psi} Y(\psi)$ be a total section, and let
	$a : \prod_{\phi:\Phi} Y(i\phi)$ be a partial section that agrees with $b$
	up to homotopy, $e : \prod_{\phi:\Phi}\, a(\phi) = b(i\phi)$.
	Then there is a total section $a' : \prod_{\psi:\Psi} Y(\psi)$ with
	$a' \circ i \seq a$, together with a homotopy
	$e' : \prod_{\psi:\Psi}\, a'(\psi) = b(\psi)$ with $e' \circ i \seq e$;
	in
	extension-type notation, $a':\Ext{i}{\Pi_\Psi Y}{a}$ and
	$e':\Ext{i}{\prod_{\psi}(a'\psi=b\psi)}{e}$.

	In particular, for a constant family: given $f : \Phi \to Y$ and
	$g : \Psi \to Y$ with $Y$ fibrant such that the triangle on the left
	commutes weakly, witnessed by
	$H : \prod_{\phi:\Phi}\, f(\phi) = g(i\phi)$, there is $h : \Psi \to Y$
	such that, in the diagram on the right,
	\[
		\begin{tikzcd}[column sep=small, row sep=small]
			\Phi  \ar[r,"f"] \ar[d,hookrightarrow,"i"]
			& Y \\
			\Psi \ar[ur,"g"',bend right] &
		\end{tikzcd}
		\qquad\rightsquigarrow\qquad
		\begin{tikzcd}[column sep=small, row sep=small]
			\Phi  \ar[r,"f"] \ar[d,hookrightarrow,"i"]
			& Y  \\
			\Psi \ar[ur,"g"',bend right] \ar[ur,"h"]&
		\end{tikzcd}
	\]
	the top triangle commutes strictly, $h \circ i \seq f$, and the bottom
	\enquote{globe} commutes weakly, witnessed by
	$G : \prod_{\psi:\Psi}\, h(\psi) = g(\psi)$, in such a way that the
	restriction of the globe to $\Phi$ is \emph{strictly} the original
	filler: $G \circ i \seq H$.
\end{lemma}
\begin{proof}
	In RS's style:
	$\Ext{i}{\prod_{\psi}\sum_{y:Y(\psi)}(y=b\psi)}{\lambda\phi.(a\phi,e\phi)}$
	is contractible by \cref{lem:rs46}, hence inhabited, and \cref{lem:rs43}
	extracts $a'$ and $e'$; note that no cofibrancy of $\Phi$ or $\Psi$ is
	needed. Alternatively, for \emph{cofibrant} $\Phi$, consider the family
	of based identity types
	$Y'(\psi) \colonequiv \Sigma(y : Y(\psi)).\,(y = b(\psi))$, a fibrant
	family with each $Y'(\psi)$ contractible. The pair
	$\lambda\phi.\,(a(\phi),\, e(\phi))$ is a partial section of $Y'$ along
	$i$. The type of its \emph{homotopy} extensions is inhabited (any two
	elements of the contractible $\prod_{\phi:\Phi} Y'(i\phi)$ are equal),
	so by \cref{lem:weak=strict-dep} the type of \emph{strict} extensions
	is inhabited as well; its components are $a'$ and $e'$.
	The constant-family instance is
	$b \colonequiv g$, $a \colonequiv f$, $e \colonequiv H$,
	$h \colonequiv a'$, $G \colonequiv e'$.
\end{proof}

\begin{longversion}
The contrast with RS is instructive: RS \emph{assume} relative funext to
derive their Proposition 4.10, whereas here it is a theorem; and the
alternative proof needs no extensionality beyond the agreement of the
strict and homotopy fibers of
$i^*$ (\cref{lem:strict-vs-homotopy-fibre}). The weak-to-strict content of RS
\S4.4 is thus free in 2LTT.
\end{longversion}

\begin{lemma}[\flink{Lemma-3-15-general} RS 4.12: $n$-types]\label{lem:rs412}
	If each $A(\psi)$ is an $n$-type, then $\Ext{i}{\Pi_\Psi A}{a}$ is an $n$-type.
\end{lemma}
\begin{proof}
	Induction on $n\ge-2$. For $n=-2$ this is \cref{lem:rs46}. For the step, given
	$f,g:\Ext{i}{\Pi_\Psi A}{a}$, \cref{lem:rs48} gives
	$(f=g)\simeq\Ext{i}{\prod_{\psi}(f\psi=g\psi)}{\lambda\phi.\mathsf{refl}}$;
	each $f\psi=g\psi$ is an $(n-1)$-type, so by the inductive hypothesis this
	extension type, hence $f=g$, is an $(n-1)$-type. Thus $\Ext{i}{\Pi_\Psi A}{a}$ is
	an $n$-type.
\end{proof}

\section{Gluing}\label{sec:gluing}

In this section and the next, the framework is instantiated to the cubical
setting. We first study
the central cubical type former, gluing.

Gluing can be formalized at several levels, from a literal rendering of the
syntactic rules to purely structural conditions; the differences in strength
carry real content. We consider four definitions:
\begin{enumerate}[noitemsep]
	\item the \emph{literal CCHM rules}
	(\cref{def:literally-translated-glue-rules}): the rules of
	\citet[Fig.~4]{CCHM2018}, internalized operation by operation, with the
	judgmental equalities rendered as strict ones;
	\item \emph{Glue structures} (\cref{def:glue-structure}): a glued family
	with strict type boundary and an $\unglue$ equivalence, strictly
	coherent with the input equivalence (a rule-style repackaging, the
	\emph{homotopy form} of the rules, is identified in
	\cref{lem:homotopy-vs-semantic});
	\item \emph{contractible Glue data} (\cref{def:strong-glue}): the type
	of Glue data is contractible, i.e.\ gluing exists uniquely;
	\item the \emph{weak Glue structure} (\cref{def:weak-glue}): a glued
	family for a single interval (no cofibration assumed), with boundary agreement only
	by inner equalities in the universe.
\end{enumerate}
Definitions (3) and (4) appear in \cref{sec:glue-ua}: they are the two ends
through which gluing is compared with univalence. The relations between
these definitions are as follows. Definitions (1) and (2) differ in exactly
two ways, both measured by one canonical comparison map $\Theta$
(\cref{lem:constructor-package}): the element-level rules of (1) are
equivalent to a \emph{strict} two-sided inverse of $\Theta$, with no
cofibrancy involved, whereas over cofibrant shapes (2) inverts $\Theta$ up
to homotopy; and (1) yields invertibility of $\unglue$ only at the level of
sections, whereas (2) demands it pointwise. Downwards, (3) implies (2), which,
instantiated at the endpoint inclusion of an interval, implies (4)
(\cref{lem:glue-strength-chain}). Univalence closes the circle: it implies
(3) for every cofibration (\cref{thm:ua-implies-strong-glue}), and (4),
for a path interval, already implies univalence
(\cref{thm:weak-glue-implies-ua}). The resulting cycle makes univalence,
(3), (4), and (2) along the interval cofibration
all equivalent (\cref{thm:glue-sandwich}.

\subsection{CCHM Glue}

\Cref{fig:cchm} is directly taken from \citet[Figure 4]{CCHM2018}; only the rule names are
added.

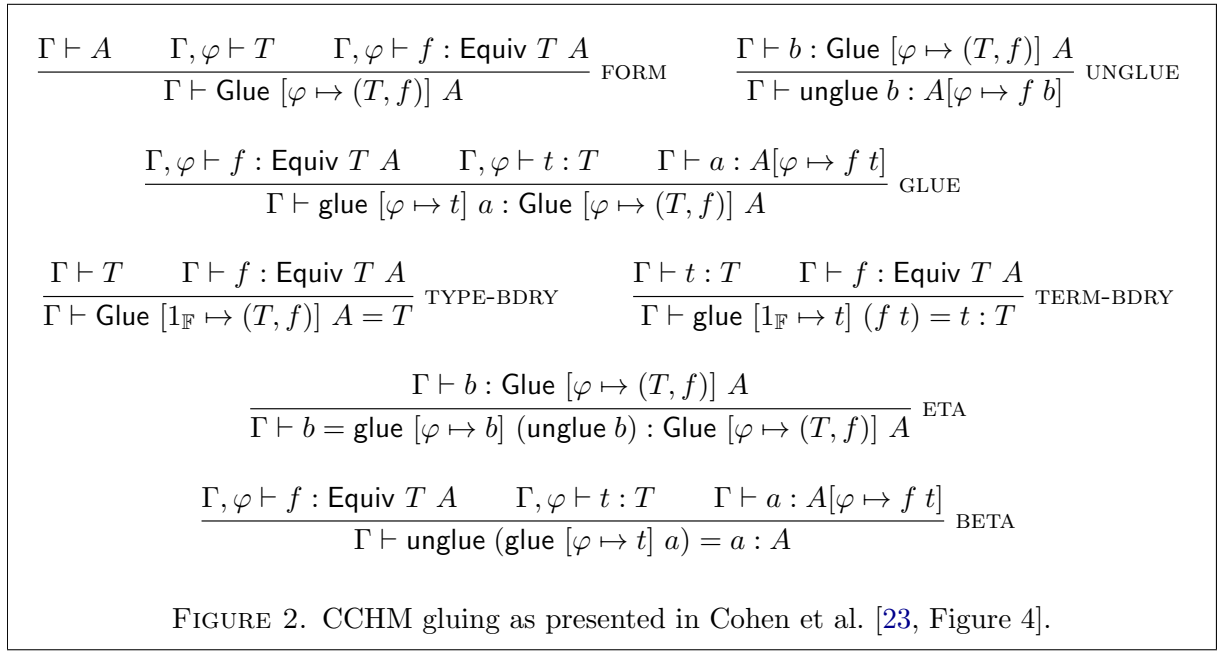
\vspace*{.5cm}

\begin{mdframed}
	\begin{mathpar}\label{figure:glue-rules-from-CCHM}
		\inferrule*[right=\rulename{form}]
		{\Gamma \vdash A \\ \Gamma,\varphi \vdash T \\
			\Gamma,\varphi \vdash f : {\sf Equiv}~T~A}
		{\Gamma \vdash {\sf Glue}~[\varphi \mapsto (T,f)]~A}
		\and
		\inferrule*[right=\rulename{unglue}]
		{\Gamma \vdash b : {\sf Glue}~[\varphi \mapsto (T,f)]~A}
		{\Gamma \vdash {\sf unglue} \; b : A[\varphi \mapsto f \;
			b]}
		\and
		\inferrule*[right=\rulename{glue}]
		{\Gamma,\varphi \vdash f : {\sf Equiv}~T~A \\
			\Gamma,\varphi \vdash t : T \\
			\Gamma \vdash a : A[\varphi \mapsto f \; t]}
		{\Gamma \vdash {\sf glue}~[\varphi \mapsto t]~a :
			{\sf Glue}~[\varphi \mapsto (T,f)]~A}

		\and
		\inferrule*[right=\rulename{type-bdry}]
		{\Gamma \vdash T \\ \Gamma \vdash f : {\sf Equiv}~T~A}
		{\Gamma \vdash {\sf Glue}~[1_\mathbb{F} \mapsto (T,f)]~A = T}
		\and
		\inferrule*[right=\rulename{term-bdry}]
		{\Gamma \vdash t : T\\
			\Gamma \vdash f : {\sf Equiv}~T~A}
		{\Gamma \vdash {\sf glue}~[1_\mathbb{F} \mapsto t]~(f \; t) = t : T}
		\and
		\inferrule*[right=\rulename{eta}]
		{\Gamma \vdash b : {\sf Glue}~[\varphi \mapsto (T,f)]~A}
		{\Gamma \vdash b = {\sf glue}~[\varphi \mapsto b]~({\sf unglue} \; b) :
			{\sf Glue}~[\varphi \mapsto (T,f)]~A}
		\and
		\inferrule*[right=\rulename{beta}]
		{\Gamma,\varphi \vdash f : {\sf Equiv}~T~A \\
			\Gamma, \varphi \vdash t : T \\
			\Gamma \vdash a : A[\varphi \mapsto f \; t ]}
		{\Gamma \vdash {\sf unglue} \; ({\sf glue}~[\varphi\mapsto t]~a) = a : A}
	\end{mathpar}
	\captionof{figure}{CCHM gluing as presented in \citet[Figure 4]{CCHM2018}.}\label{fig:cchm}
\end{mdframed}

\vspace*{.5cm}

Our first step is to internalize this.
In our setting, the annotated type $A[\varphi \mapsto a]$ is exactly the
extension type $\Ext{i}{\Pi_\Gamma A}{a}$ of \cref{def:ext-formal}.
The annotation $A[\varphi \mapsto f\,b]$ is subtler: it applies $f$ to the
restriction of $b : \Glue$, which is an element of $T$ only by virtue of the
boundary rule \rulename{type-bdry}. In \cref{fig:cchm}
this is silent, the equality being judgmental; internally it forces us to
introduce the strict boundary of $\Glue$ \emph{before} any rule that mentions
elements, and to coerce along it. This dictates the order of the items in the definition below.

\begin{definition}[\flink{Definition-4-1} The CCHM rules on sections]\label{def:literally-translated-glue-rules}
	Let $\UU$ be a fibrant universe. We say that $\UU$ \emph{has CCHM Glue on sections} if we
	have the operations (1)--(7) below (one for each of the seven rules of
	\cref{fig:cchm}, in a different order) for every cofibration
	$i : \Phi \hookrightarrow \Gamma$ (here and throughout the gluing
	development we write $\Gamma$ for the codomain of the cofibration,
	matching the ambient context of \cref{fig:cchm}; the extension-type sections
	keep $\Psi$) and all data
	\[
	A:\Gamma\to\UU,
	\qquad
	T:\Phi\to\UU,
	\qquad
	f:\prod_{\phi:\Phi} T(\phi)\simeq A(i\phi).
	\]
	(We suppress the dependence of the operations on $(i,A,T,f)$ except in the
	name of the glued family itself.)
	\begin{enumerate}
		\item \emph{(Formation; rule \rulename{form}.)} A family
		$\Glue^{T,f,A} : \Gamma\to\UU$.
		\item \emph{(Type boundary; rule \rulename{type-bdry}.)} A strict equality
		\[
		p \,:\, \Glue^{T,f,A}\circ i \,\seq\, T.
		\]
		In \cref{fig:cchm}, this rule is stated at the top face formula
		$1_\mathbb{F}$, i.e.\ the instance $i=\id$. By stating it with
		a general cofibration, we add the strict
		stability of the $\Glue$ former under the substitution $i$, which
		in \citet{CCHM2018} is captured by the ambient presheaf rather than \cref{fig:cchm}.
		\item \emph{(unglue; rule \rulename{unglue}.)} A function
		\[
		\unglue \,:\, \prod_{\gamma:\Gamma}\, \Glue^{T,f,A}(\gamma) \to A(\gamma)
		\]
		such that, for every $\phi:\Phi$, the function $\unglue_{i\phi}$ is strictly
		equal to $f^p_\phi$, where $f^p_\phi : \Glue^{T,f,A}(i\phi)\to A(i\phi)$
		denotes $f_\phi$ with its source coerced along the type boundary $p$ of (2).
		This strict equality internalizes the boundary annotation
		$A[\varphi\mapsto f\,b]$ at the \emph{generic} element $b$.
		Consequently, for a section
		$b : \prod_{\gamma:\Gamma}\Glue^{T,f,A}(\gamma)$, the composite
		$\unglue \odot b \colonequiv \lambda\gamma.\,\unglue_\gamma(b\,\gamma)$
		carries a strict boundary witness and is an element of
		$\Ext{i}{\Pi_\Gamma A}{\;\lambda\phi.\,f^p_\phi(b(i\phi))}$,
		the literal reading of \rulename{unglue} in the ambient context
		$\Gamma$.
		\item \emph{(glue; rule \rulename{glue}.)} For every partial section
		$t : \prod_{\phi:\Phi} T(\phi)$ and every
		$a : \Ext{i}{\Pi_\Gamma A}{\;\lambda\phi.\,f_\phi(t\,\phi)}$, an
		element
		\[
		\glue(t,a) \,:\, \prod_{\gamma:\Gamma} \Glue^{T,f,A}(\gamma).
		\]
		We write $a_0 : \prod_{\gamma:\Gamma}A(\gamma)$ for the underlying section
		of $a$ (\cref{def:ext-formal}).
		\item \emph{(Boundary of glue; rule \rulename{term-bdry}.)} In the situation of (4), and coerced along $p$,
		\[
		\glue(t,a)\circ i \,\seq\, t.
		\]
		The same comment as for \rulename{type-bdry} applies: in \cref{fig:cchm}, \rulename{term-bdry} is stated at $1_\mathbb{F}$ only;
		taking into account the (silent) stability under substitution, we get this general form.
		\item \emph{($\eta$; rule \rulename{eta}.)} For every
		$b : \prod_{\gamma:\Gamma}\Glue^{T,f,A}(\gamma)$,
		\[
		b \,\seq\, \glue\bigl(b\circ i,\ \unglue\odot b\bigr),
		\]
		where $b\circ i$ is regarded as an element of $\prod_\Phi T$ via $p$,
		and $\unglue\odot b$ carries the strict boundary witness of item (3),
		which matches the annotation required in item (4) since
		$(\unglue\odot b)\circ i \seq \lambda\phi.\,f^p_\phi(b(i\phi))$.
		\item \emph{($\beta$; rule \rulename{beta}.)} In the situation of (4),
		\[
		\unglue\odot\glue(t,a) \,\seq\, a_0.
		\]
	\end{enumerate}
\end{definition}

Items (1)--(3) are generic: a family, a strict equality of families, and a map
given at every $\gamma:\Gamma$. Items (4)--(7) are asserted in the ambient
context $\Gamma$, where all their premises are sections; the rules of \cref{fig:cchm}
apply after every substitution into $\Gamma$ as well, and we do not demand that.
Their generic form is not available in the vocabulary of families over $\Phi$
and $\Gamma$ and precomposition with $i$. A partial element over
$\gamma:\Gamma$ is a section of $T$ over the strict fiber
$\sfib{i}(\gamma) \equiv \Sigma(\phi:\Phi).\,(i\phi\seq\gamma)$, and those data do not
provide it.

We call items (1)--(3) of \cref{def:literally-translated-glue-rules},
together with the rule that
$\unglue_\gamma:\Glue^{T,f,A}(\gamma)\to A(\gamma)$ is an equivalence for
every $\gamma:\Gamma$, the \emph{homotopy form} of the Glue rules; up to a
canonical adjustment of the $\mathsf{isEquiv}$ witnesses, it is precisely a
Glue structure as defined in the next subsection (cf.~\cref{lem:homotopy-vs-semantic} below).

The two rule packages are related by a canonical comparison map. For an
input datum $(A,T,f)$ and $t:\prod_\Phi T$, write
$f\cdot t\colonequiv\lambda\phi.\,f_\phi(t\,\phi):\prod_\Phi A\circ i$.

\begin{lemma}[\flink{Lemma-4-2-i-to} The constructor package inverts a comparison map]\label{lem:constructor-package}
	Let $\UU$ be a fibrant universe, $i:\Phi\hookrightarrow\Gamma$ a cofibration, $(A,T,f)$ an input datum, and
	$(G,p,\unglue)$ with $G\colonequiv\Glue^{T,f,A}$ data as in items
	(1)--(3) of \cref{def:literally-translated-glue-rules}. Define the
	\emph{restriction--unglue map}
	\[
		\Theta \,:\, \prod_{\gamma:\Gamma}G(\gamma)
		\;\longrightarrow\;
		\sum_{t:\prod_\Phi T}\Ext{i}{\Pi_\Gamma A}{f\cdot t},
		\qquad
		\Theta(b)\colonequiv\bigl(b\circ i,\;(\unglue\odot b,\,w_b)\bigr),
	\]
	where $w_b$ is the strict boundary witness supplied by item (3) and coercion by the type boundary equation (2) is implicit.
	\begin{enumerate}
		\item The element-level rules, items (4)--(7) of
		\cref{def:literally-translated-glue-rules}, hold for this input
		if and only if $\Theta$ admits a strict two-sided inverse (which is
		then $\glue$ itself). This is a purely formal reformulation; no
		cofibrancy is used.
		\item Suppose $\Phi$ is cofibrant (hence so is $\Gamma$) and each
		$\unglue_\gamma$ is an equivalence, as in the homotopy form. Then
		$\Theta$ is an equivalence
		of fibrant types. Consequently the constructor package holds up to
		homotopy: there is $\glue(t,a):\prod_\Gamma G$ with \emph{strict}
		boundary $\glue(t,a)\circ i\seq t$ (\rulename{term-bdry}) and with
		inner equalities $\Theta(\glue(t,a))=(t,a)$ and $\glue(\Theta\,b)=b$
		(\rulename{beta} and \rulename{eta} in linked form; the
		componentwise laws follow by projecting). Disregarding the strict
		boundary, the pair consisting of $\glue(t,a)$ and the inner equality
		$\Theta(\glue(t,a))=(t,a)$ is contractibly unique, being an element
		of the homotopy fiber of the equivalence $\Theta$ over $(t,a)$.
		\item If $\Phi$ is cofibrant and $\Theta$ admits a strict two-sided
		inverse, then $\unglue\odot(-):\prod_\Gamma G\to\prod_\Gamma A$ is an
		equivalence: invertibility at the level of sections, without any
		assumption of pointwise invertibility of $\unglue$.
	\end{enumerate}
	In short,
	\[
		\begin{aligned}
			&\text{homotopy form} \;+\; \text{strict two-sided inverses of }\Theta\\
			&\qquad=\;\; \text{literal rules} \;+\; \text{pointwise invertibility of }\unglue.
		\end{aligned}
	\]
\end{lemma}
\begin{proof}[Proof sketch]
	Item (1) is bookkeeping: items (5) and (7) make the two components of
	$\Theta(\glue(t,a))$ strictly $t$ and $a_0$, the remaining witnesses inhabit
	strict propositions and agree by outer UIP, and item (6) is
	$\glue(\Theta\,b)\seq b$; conversely the component equations of a strict
	two-sided inverse of $\Theta$ are exactly items (5)--(7). For item (2), a
	fiberwise equivalence over a cofibrant shape induces an equivalence on
	products, so $\unglue\odot(-)$ is an equivalence; that composite is
	$\kappa\circ\Theta$ for the projection $\kappa$ out of the $\Sigma$-type,
	whose fibers are contractible, so $\Theta$ is an equivalence by
	2-out-of-3, and the constructor package is an element of its homotopy fiber
	over $(t,a)$. Item (3) follows because a strict isomorphism of fibrant types
	is an equivalence. The argument is carried out in full in
	\texttt{Extension.GlueConstructorPackage}.
\end{proof}

\subsection{Semantic gluing}

\begin{definition}[\flink{Definition-4-3-data} Glue structure]\label{def:glue-structure}
	Let \(i:\Phi \hookrightarrow \Gamma\) be a cofibration and $\UU$ a
	fibrant universe. A \emph{Glue structure for $(i,\UU)$} assigns, to every
	input triple $(A, T, e)$ of types
	\[
	A:\Gamma\to\UU,
	\qquad
	T:\Phi\to\UU,
	\qquad
	e:\prod_{\phi:\Phi} T(\phi)\simeq A(i\phi),
	\]
	a triple $(G, p, u)$ of types
	\[
	G:\Gamma\to\UU,
	\qquad
	p:G\circ i\seq T,
	\qquad
	u:\prod_{\gamma:\Gamma} G(\gamma)\simeq A(\gamma)
	\]
	such that $u$ restricts to $e$ up to strict equality: for every
	$\phi:\Phi$, the composite equivalence
	$T(\phi) \seq G(i\phi) \simeq A(i\phi)$ is strictly equal to $e_\phi$
	(the \emph{coherence}).
	For a \emph{fixed} input triple $(A,T,e)$, we write
	$\mathsf{GlueData}_{i,\UU}(A,T,e)$ for the type of triples $(G,p,u)$
	together with their coherence witnesses; a Glue structure is thus an
	element of the product of the types $\mathsf{GlueData}_{i,\UU}(A,T,e)$
	over all inputs.
\end{definition}

Formulated as diagrams, the above definition says the following for a given cofibration \(i:\Phi \hookrightarrow \Gamma\) and $\UU$:
whenever the triangle on the left commutes pointwise up to equivalence, there is $G$ such that, in the diagram on the right,
the top triangle commutes strictly and the bottom ``globe'' commutes up to equivalence,
in such a way that the fillers ``compose'' to the original filler up to strict equality (cf.~\cref{lem:rs410}).
\[
	\begin{tikzcd}[column sep=small, row sep=small]
		\Phi  \ar[r,"T"] \ar[d,hookrightarrow,"i"]
		& \UU  \\
		\Gamma \ar[ur,"A"',bend right] & 
	\end{tikzcd}
	\qquad\rightsquigarrow\qquad
	\begin{tikzcd}[column sep=small, row sep=small]
		\Phi  \ar[r,"T"] \ar[d,hookrightarrow,"i"]
		& \UU  \\
		\Gamma \ar[ur,"A"',bend right] \ar[ur,"G"]& 
	\end{tikzcd}
\]

The coherence of \cref{def:glue-structure} compares the
restriction of $u$ with $e$ as \emph{equivalences}; here and below, we write
$e^{p}_\phi$ for $e_\phi$ with its source coerced to $G(i\phi)$ along $p$,
so that the coherence reads $u_{i\phi}\seq e^{p}_\phi$.

\begin{remark}[\flink{Remark-4-4} Homotopy form $=$ Glue structure]\label{lem:homotopy-vs-semantic}\label{def:homotopy-glue-rules}
	The coherence is a priori stronger than strict agreement of the
	\emph{underlying functions}, since the $\mathsf{isEquiv}$ witnesses are
	only propositionally unique. But only a priori: by the realignment
	principle \cref{cor:prop-realign}, the witnesses can be adjusted,
	contractibly uniquely and changing nothing else. Consequently, with no
	cofibrancy assumptions at all, $(i,\UU)$ satisfies the Glue rules in
	homotopy form if and only if it has a Glue structure.
	Indeed, items (1) and (2) of \cref{def:literally-translated-glue-rules}
	are literally the data $(G,p)$ of a Glue structure; a coherent
	$(G,p,u)$ yields the homotopy form by taking underlying maps,
	$\unglue\colonequiv\pi_1\circ u$, with invertibility witnessed by
	$\pi_2\circ u$; and homotopy-form data $(G,p,\unglue)$ with
	invertibility proofs $\varepsilon$ realign to a coherent
	$(G,p,u)$ with $\pi_1\circ u\seq\unglue$, by \cref{cor:prop-realign}
	applied to the data family
	$D(\gamma)\colonequiv\bigl(G(\gamma)\to A(\gamma)\bigr)$ with property
	$\mathsf{isEquiv}$, total section $(\unglue,\varepsilon)$ and partial
	section $e^{p}$: the strict boundary of item (3) is exactly the
	required strict agreement of the underlying data.
\end{remark}

One might hope that having CCHM Glue in the literal sense of
\cref{def:literally-translated-glue-rules} for all cofibrations is
equivalent to having Glue structures for all cofibrations in the sense of
\cref{def:glue-structure}. The two notions differ, in both directions. From
literal to semantic: the rules of \cref{def:literally-translated-glue-rules}
constrain only \emph{sections} of the glued family, whereas
\cref{def:glue-structure} demands that $\unglue$ be a \emph{pointwise}
equivalence, a statement that in \citet[\S7.2]{CCHM2018} is derived from the
composition structures of the constituent types, which are not part of \cref{fig:cchm}. From
semantic to literal: a Glue structure offers no strict $\eta$- and
$\beta$-laws. \Cref{lem:constructor-package,lem:homotopy-vs-semantic}
delimit exactly what does hold.

\begin{longversion}
The following remark explains why the two gaps are genuine.

\begin{remark}[The two gaps are genuine]\label{rem:connection-gaps}
	The two discrepancies separating the literal rules from the semantic
	Glue structures cannot be removed.
	\begin{enumerate}
		\item \rulename{eta} and \rulename{beta} cannot be upgraded from
		fibrant to strict equality by these means: by
		\cref{lem:constructor-package}(1), a $\glue$ satisfying items
		(5)--(7) strictly is precisely a strict two-sided inverse of
		$\Theta$, and producing one from the mere equivalence of
		\cref{lem:constructor-package}(2) is a \emph{realignment}
		(strictification) principle in the style of Orton--Pitts
		\citep[Sec.~6]{OrtonPitts2018}, which bare 2LTT does not provide.
		Compare the explicit pointwise candidate
		$\sum_{a:A}\prod_{p:\varphi}\sum_{t:T(p)}\bigl(e_p(t)\seq a\bigr)$:
		it satisfies the strict boundary laws by construction but is not
		known to be fibrant, while the semantic glued family is fibrant but
		satisfies the laws only up to homotopy.
		\item Pointwise invertibility of $\unglue$, the one ingredient
		of \cref{def:glue-structure} that the literal rules do not deliver,
		is not an artifact of the translation: the rules constrain only
		sections of the glued family, and accordingly they yield
		invertibility of $\unglue$ only at the level of sections
		(\cref{lem:constructor-package}(3)). In
		\citet[Secs.~6--7]{CCHM2018}, the statement that $\unglue$ is an
		equivalence is derived from the \emph{composition structure} of
		${\sf Glue}$, which is not part of \cref{fig:cchm}.
	\end{enumerate}
\end{remark}
\end{longversion}

\begin{longversion}
\paragraph{A fibered, universe-free formulation.}
The gluing operation can also be formulated for \emph{fibrations} over the
shapes instead of families into a universe. This is the form in which the
\emph{equivalence extension property} appears in the semantic literature
\citep[Sec.~5]{Sattler2017}; the internal, strictified version is discussed
in \cref{rem:strong-glue-eep} below. For a fibration
$q:\widetilde{Y}\twoheadrightarrow\Gamma$ we write
$i^*\widetilde{Y}\twoheadrightarrow\Phi$ for its pullback along $i$, and we
call a map over a shape a \emph{fiberwise equivalence} if it induces an
equivalence on each strict fiber (the strict fibers of a fibration are
fibrant \citep[Lem.~3.9]{2LTT}, so this is meaningful).

\begin{definition}[Fibered Glue structure]\label{def:fibred-glue}
	Let $i:\Phi\hookrightarrow\Gamma$ be a cofibration. A \emph{fibered
	Glue structure} for $i$ assigns, to every input (the black part of the
	diagram below) consisting of fibrations
	$\widetilde{T}\twoheadrightarrow\Phi$ and
	$\widetilde{A}\twoheadrightarrow\Gamma$ together with a fiberwise
	equivalence $\widetilde{e}:\widetilde{T}\to i^*\widetilde{A}$ over
	$\Phi$, an output (the green part) consisting of a fibration
	$\widetilde{G}\twoheadrightarrow\Gamma$, a strict isomorphism
	$i^*\widetilde{G}\siso\widetilde{T}$ over $\Phi$ (the
	green pullback square), and a fiberwise equivalence
	$\unglue:\widetilde{G}\to\widetilde{A}$ over $\Gamma$ whose restriction
	along $i$ agrees with $\widetilde{e}$ modulo the pullback
	identification: up to fiberwise homotopy for the \emph{homotopy
	coherent} version, up to strict equality for the \emph{strictly
	coherent} one.
	\begin{center}
	\begin{tikzpicture}[>=stealth, scale=1.1]

		\node (I)    at (0,2.5) {$\Phi$};
		\node (T)    at (3,2.5) {$\widetilde{T}$};
		\node (Atop) at (1.6,1.9) {$i^*\widetilde{A}$};

		\node (J)    at (0,0) {$\Gamma$};
		\node[text=green!60!black] (Glue) at (3,0) {$\widetilde{G}$};
		\node (Abot) at (1.6,-0.6) {$\widetilde{A}$};

		\draw[->]          (I)    -- node[left] {$i$} (J);
		\draw[->]          (Atop) -- (Abot);
		\draw[->,green!60!black] (T)    -- (Glue);

		\draw[<<-] (I) -- (Atop);
		\draw[<<-] (I) -- (T);
		\draw[->] (T) -- node[below] {$\widetilde{e}$} (Atop);

		\draw[<<-]               (J)    -- (Abot);
		\draw[<<-, dashed,green!60!black]       (J)    -- (Glue);
		\draw[->,green!60!black] (Glue) -- node[below right] {$\unglue$} (Abot);

		\draw (Atop) ++(-0.05,-0.4) -- ++(-0.4,+0.12) -- ++(0,+0.4);
		\draw[green!60!black] (T) ++(-0.05,-0.4) -- ++(-0.4,-0.22) -- ++(0,+0.4);
	\end{tikzpicture}
	\end{center}
\end{definition}

\begin{lemma}[Family gluing induces fibered gluing]\label{lem:family-vs-fibred}
	Let $i:\Phi\hookrightarrow\Gamma$ be a cofibration and $\UU$ a
	fibrant universe. A Glue structure for $(i,\UU)$
	(\cref{def:glue-structure}) induces a strictly coherent fibered Glue
	structure for all inputs whose strict fibers lie in $\UU$.
\end{lemma}
\begin{proof}
	Let $(\widetilde{T},\widetilde{A},\widetilde{e})$ be such an input. By
	\citet[Lem.~3.9]{2LTT}, the strict fibers form fibrant families
	$T:\Phi\to\UU$ and $A:\Gamma\to\UU$, with canonical strict isomorphisms
	$\widetilde{T}\siso\sum_\Phi T$ over $\Phi$ and
	$\widetilde{A}\siso\sum_\Gamma A$ over $\Gamma$; under these,
	$\widetilde{e}$ induces a family of equivalences
	$e:\prod_{\phi:\Phi}T(\phi)\simeq A(i\phi)$. Let $(G,p,u)$ be the
	output of the given Glue structure at $(A,T,e)$, and set
	$\widetilde{G}\colonequiv\sum_\Gamma G$, a fibration over $\Gamma$. Then
	\[
		i^*\widetilde{G}
		\;\siso\;
		\textstyle\sum_\Phi G\circ i
		\;\seq\;
		\textstyle\sum_\Phi T
		\;\siso\;
		\widetilde{T}
		\qquad\text{over }\Phi,
	\]
	using $p$ in the middle step; and the total map
	$(\gamma,g)\mapsto(\gamma,\,u_\gamma(g))$, composed with the canonical
	isomorphism $\sum_\Gamma A\siso\widetilde{A}$, is a
	fiberwise equivalence $\unglue:\widetilde{G}\to\widetilde{A}$ over
	$\Gamma$. The strict coherence of $u$ transfers to the strict boundary
	agreement of $\unglue$ under these identifications.
\end{proof}
\end{longversion}

\begin{longversion}
\begin{remark}[Universe realignment]\label{rem:universe-realignment}
	A \emph{realignment property of the universe} along $i$ asks: given
	$G_0:\Gamma\to\UU$ and pointwise strict isomorphisms
	$G_0(i\phi)\siso T(\phi)$ on $\Phi$, find $G:\Gamma\to\UU$, strictly
	isomorphic to $G_0$ over $\Gamma$, with $G\circ i\seq T$.
	Cubical universes are constructed so as to satisfy such realignment;
	see the strictness axiom of \citet[Sec.~6]{OrtonPitts2018} and the
	alignment step in the universe constructions of \citet{ABCFHL2021}.
	Bare 2LTT provides nothing of the sort. This is what separates the
	family form of \cref{def:glue-structure} from fibration-based,
	universe-free formulations of gluing, as in the semantic equivalence
	extension property (cf.\ \cref{rem:strong-glue-eep}): passing to strict
	fibers of a glued fibration yields all of a Glue structure \emph{except}
	the boundary condition, strict \emph{isomorphisms}
	$G(i\phi)\siso T(\phi)$ rather than the strict \emph{equality}
	$G\circ i\seq T$, and strictly isomorphic elements of a universe need
	not be strictly equal.

	The difference matters: the converse direction of the
	univalence--gluing equivalence (\cref{lem:glue-strength-chain}(2) and
	\cref{thm:weak-glue-implies-ua} below) consumes the strict equality,
	reading off an \emph{inner equality}
	$A=_\UU B$ from a line $G:I\to\UU$ with $G(0)\seq A$ and $G(1)\seq B$.
	With only strict isomorphisms $G(0)\siso A$ and
	$G(1)\siso B$, one would first have to convert an
	isomorphism into an inner equality in $\UU$, a special case of univalence, so
	the argument would be circular. A fibration-based form therefore cannot
	replace the family form of \cref{def:glue-structure} in what follows.

	We note a pattern. The comparison results of this section turn on
	three independent \emph{realignment} demands: for the propositional
	$\mathsf{isEquiv}$ witnesses (\cref{cor:prop-realign}; automatic, by
	relative funext), for the constructor package (strict two-sided
	inverses of $\Theta$, \cref{lem:constructor-package}(1); not automatic),
	and for the universe (the present remark; not automatic).
\end{remark}
\end{longversion}

\begin{longversion}
\subsection{Explicit pointwise formulas: homotopy and strict fibers}\label{sec:strict-glue-from-ua}

The preceding definitions specify gluing by rules, or as structure on
$(i,\UU)$. We now record two explicit \emph{candidate formulas}, both
pointwise over a proposition $\varphi$. The first uses homotopy fibers; it
is essentially the non-strict gluing type of Orton--Pitts
\citep[Def.~6.1]{OrtonPitts2018}, requires no 2LTT assumptions at all, and
univalence alone makes it a Glue object up to inner equality in the universe
(\cref{rem:pointwise-glue-from-UA}, the baseline for the forward direction
of \cref{sec:glue-ua}). The second replaces the homotopy fibers by strict
fibers, with the opposite profile: strict boundary laws by construction,
fibrancy only under extra hypotheses.

\begin{definition}[Pointwise Glue]\label{def:pointwise-glue}
	Let \(\varphi : \Prop\) be a proposition, \(A : \UU\), \(T : \varphi \to \UU\),
	and \(e : \prod_{p:\varphi} T(p) \simeq A\) a family of equivalences. Define
	\[
		\mathsf{Glue}_{\mathrm{pt}}(\varphi,T,e,A)
		\;\colonequiv\;
		\sum_{a:A}\prod_{p:\varphi}\hfib{e_p}(a),
		\qquad
		\hfib{e_p}(a) \colonequiv \sum_{t:T(p)}(e_p(t)=a),
	\]
	together with
	\(\mathsf{unglue} : \mathsf{Glue}_{\mathrm{pt}}(\varphi,T,e,A) \to A\),
	\(\mathsf{unglue}(a,h) \colonequiv a\).
\end{definition}

\begin{remark}[Strict-fiber variant]
One can replace the homotopy fiber in the pointwise Glue
(\cref{def:pointwise-glue}) by a strict fiber and define
\[
  \mathsf{Glue}_{\mathrm{strict}}(\varphi,T,e,A)
  \;\colonequiv\;
  \sum_{a:A}\prod_{p:\varphi}\sfib{e_p}(a),
  \qquad
  \sfib{e_p}(a)
  \colonequiv
  \sum_{t:T(p)}(e_p(t)\seq a).
\]
This immediately gives strict compatibility on the boundary: if
\(x=(a,h):\mathsf{Glue}_{\mathrm{strict}}(\varphi,T,e,A)\) and
\(p:\varphi\), then
\[
  e_p(\pi_1(h(p))) \seq \mathsf{unglue}(x).
\]
However, this stricter definition is not a consequence of univalence alone.
Univalence turns equivalences into inner equalities in the universe; it does not turn the
homotopies witnessing an equivalence \(e_p:T(p)\simeq A\) into strict equalities
\(e_p(t)\seq a\). Thus the strict-fiber version is generally too small unless
the boundary maps are already strict equivalences, or unless one assumes an
additional strictification/realignment principle. In particular, without such extra
hypotheses, \(\mathsf{Glue}_{\mathrm{strict}}\) is not known to be fibrant:
strict equality and strict fibers do not preserve fibrancy in general.
\end{remark}

\begin{proposition}[Strict boundary for strict Glue]\label{prop:strict-glue-boundary}
Assume that \(\varphi\) is a strict proposition, in the sense that in any
context containing \(p:\varphi\), every \(q:\varphi\) is strictly equal to
\(p\). Then, in the boundary context \(p:\varphi\), there is a strict
isomorphism
\[
  \mathsf{Glue}_{\mathrm{strict}}(\varphi,T,e,A) \siso T(p).
\]
Consequently, if \(T(p)\) is fibrant, then
\(\mathsf{Glue}_{\mathrm{strict}}(\varphi,T,e,A)\) is fibrant in that boundary
context.
\footnote{Related to the boundary-strictness discussion in \citet[Def.~6.11, Eq.~(6.9)]{OrtonPitts2018} and the judgmental boundary rule for CCHM Glue \citep[Fig.~4]{CCHM2018}; the present pointwise strict-fiber calculation differs from Orton--Pitts' \(SGlue\), which uses their axiom ax9 to strictify gluing.}
\end{proposition}

\begin{proof}
In the context \(p:\varphi\), strict proof-irrelevance for \(\varphi\) gives a
strict isomorphism
\[
  \prod_{q:\varphi}\sfib{e_q}(a)
  \siso
  \sfib{e_p}(a).
\]
Therefore
\[
  \mathsf{Glue}_{\mathrm{strict}}(\varphi,T,e,A)
  \siso
  \sum_{a:A}\sum_{t:T(p)}(e_p(t)\seq a)
  \siso
  \sum_{t:T(p)}\sum_{a:A}(e_p(t)\seq a).
\]
For each \(t:T(p)\), the type \(\sum_{a:A}(e_p(t)\seq a)\) is strictly
contractible, with center \((e_p(t),\mathsf{refl}^{\mathrm{s}})\). Hence the
last displayed type is strictly isomorphic to \(T(p)\). Fibrancy is preserved
by strict isomorphism.
\end{proof}

\begin{proposition}[Fibrancy of strict Glue from cofibrancy]\label{prop:strict-glue-fibrant-from-cofibrancy}
Assume that \(\varphi\) is a \emph{cofibrant strict proposition},
i.e.\ both \(\szero\to[\varphi]\) and
\([\varphi]\to 1\) are cofibrations; the proof uses both halves. Assume
moreover that \(A\) is
fibrant and that \(T(p)\) is fibrant for every \(p:\varphi\). Then
\(\mathsf{Glue}_{\mathrm{strict}}(\varphi,T,e,A)\) is fibrant.
\footnote{Related in purpose to fibrancy of gluing in \citet[Thm.~6.8]{OrtonPitts2018}, which however assumes equivalence data and constructs a CCHM fibration structure; the present proposition is a separate 2LTT fibrancy argument.}
\end{proposition}

\begin{proof}
The point is to avoid proving fibrancy of each strict fiber separately. Instead,
rewrite the whole type as a pullback:
\[
  \mathsf{Glue}_{\mathrm{strict}}(\varphi,T,e,A)
  \siso
  A \times_{\prod_{p:\varphi}A} \prod_{p:\varphi}T(p).
\]
Here the map \(A\to\prod_{p:\varphi}A\) sends \(a\) to the constant family
\(\lambda p.a\), and the map \(\prod_{p:\varphi}T(p)\to\prod_{p:\varphi}A\)
sends \(t\) to \(\lambda p.e_p(t(p))\). Since \([\varphi]\to 1\) is a
cofibration and \(A\) is fibrant, the first map is a fibration. Since each
\(T(p)\) is fibrant and \([\varphi]\) is cofibrant, the dependent product
\(\prod_{p:\varphi}T(p)\) is fibrant. Pulling back the fibration
\(A\to\prod_{p:\varphi}A\) along the map from this fibrant type therefore gives
a fibrant total space. Thus \(\mathsf{Glue}_{\mathrm{strict}}\) is fibrant.
\end{proof}

The following proposition is \emph{not} subsumed by
\cref{thm:ua-implies-strong-glue}: it assumes strict equivalence data, but
needs neither univalence nor cofibrancy, and its conclusions are strict.

\begin{proposition}[Strict Glue from strict equivalences]\label{prop:strict-glue-from-strict-equivalences}
Assume the data \(\varphi,A,T\) as above, but suppose that each boundary map
\(e_p:T(p)\to A\) is equipped with a strict inverse \(r_p:A\to T(p)\), with
strict retraction and section laws
\[
  e_p(r_p(a))\seq a,
  \qquad
  r_p(e_p(t))\seq t.
\]
Then \(G^{\mathrm{s}}\colonequiv
\mathsf{Glue}_{\mathrm{strict}}(\varphi,T,e,A)\) has the following stricter
properties.
\begin{enumerate}
  \item The map \(\mathsf{unglue}:G^{\mathrm{s}}\to A\) is a strict equivalence.
        Hence, if \(A\) is fibrant, then \(G^{\mathrm{s}}\) is fibrant.
  \item For every \(p:\varphi\), the boundary projection
        \(\rho_p:G^{\mathrm{s}}\to T(p)\), defined by
        \(\rho_p(a,h)\colonequiv \pi_1(h(p))\), is a strict equivalence.
  \item The boundary compatibility is strict:
        \[
          e_p\circ \rho_p \seq \mathsf{unglue}.
        \]
\end{enumerate}
These properties still do not identify \(G^{\mathrm{s}}\) judgmentally with
\(T(p)\) under \(p:\varphi\). Under strict proof-irrelevance for \(\varphi\),
\cref{prop:strict-glue-boundary} gives a strict isomorphism on the boundary;
a judgmental identification is the kind of extra strictness axiom isolated by
Orton--Pitts.
\footnote{Related to the elementary properties of the Orton--Pitts gluing object \citep[Defs.~6.1, 6.11]{OrtonPitts2018}; we are not aware of this exact pointwise strict-fiber statement with strict equivalences as boundary maps, an elementary variant tailored to the 2LTT setting.}
\end{proposition}

\begin{proof}
The inverse to \(\mathsf{unglue}\) sends \(a:A\) to
\((a,\lambda p.(r_p(a), e_p(r_p(a))\seq a))\). The strict retraction law gives
one composite. The other composite follows because each strict fiber
\(\sfib{e_p}(a)\) is strictly contractible with center
\((r_p(a), e_p(r_p(a))\seq a)\). Hence \(\mathsf{unglue}\) is a strict
equivalence, so \(G^{\mathrm{s}}\) is strictly equivalent to \(A\) and is fibrant
whenever \(A\) is fibrant. For a fixed \(p:\varphi\), the displayed strict equality
\(e_p\circ\rho_p\seq\mathsf{unglue}\) is the second component of \(h(p)\).
Since \(e_p\) and \(\mathsf{unglue}\) are strict equivalences, so is
\(\rho_p\).
\end{proof}

\begin{remark}[Pointwise versus restriction-map strict fibers]\label{rem:pointwise-vs-restriction-strict-fibres}
The strict fibers of this subsubsection are taken \emph{pointwise}, inside
the glued type itself. They are the right tool when the boundary
equivalences are already strict
(\cref{prop:strict-glue-from-strict-equivalences}), but for ordinary
univalent equivalences they are too small, as discussed above. For
ordinary equivalences, one instead takes strict fibers of
\emph{restriction maps}: the strict fiber of the universe
restriction \((\Gamma\to\UU)\to(\Phi\to\UU)\) makes the type boundary
\(G\circ i\seq T\) strict, and the strict fiber of the restriction map on
the family of equivalences \(G(\gamma)\simeq A(\gamma)\) makes
\(\mathsf{unglue}\) agree strictly with \(e\) on the boundary. In bundled
and contractibly unique form, this is how univalence produces
contractible Glue data in \cref{thm:ua-implies-strong-glue}: the strict fiber of
the restriction of the universal family
\(W(\gamma)\colonequiv\Sigma(X:\UU).(X\simeq A(\gamma))\)
(\cref{lem:glue-data-fibrant}) performs both strictification steps at
once.
\end{remark}
\end{longversion}

\paragraph{Pointwise formulas.}
For $\varphi:\Prop$, $A:\UU$, $T:\varphi\to\UU$, and
$e:\prod_{p:\varphi}T(p)\simeq A$, assuming the displayed
$\Sigma$- and $\Pi$-types are available, the pointwise formula
\[
G_{\mathrm{pt}}
\colonequiv
\sum_{a:A}\prod_{p:\varphi}\hfib{e_p}(a)
\]
is the pointwise form of the non-strict gluing construction of
\citet[Def.~6.1]{OrtonPitts2018}. The maps
$\mathsf{unglue}(a,h)\colonequiv a$ and
$\rho_p(a,h)\colonequiv\pi_1(h(p)):T(p)$ are equivalences,
entirely in HoTT; univalence is used only to turn each $\rho_p$
into an inner equality $G_{\mathrm{pt}}=_\UU T(p)$.
Replacing $\hfib{e_p}(a)$ by $\sfib{e_p}(a)$ makes the equations
$e_p(\pi_1(h(p)))\seq a$ strict, but the resulting outer type need
not be fibrant and is not supplied by univalence alone. The
construction of \cref{lem:glue-data-fibrant,thm:ua-implies-strong-glue}
instead takes a strict fiber of a restriction map, strictly
realigning the type and unglue boundaries at once.

\section{The equivalence of gluing and univalence}\label{sec:glue-ua}

Having introduced the different formulations of gluing, we now prove that
gluing is equivalent to univalence, in the strongest form that we can. The
strategy is a cycle of implications: we isolate the strongest condition,
\emph{contractible Glue data} (contractibility of the type of Glue data,
along an \emph{arbitrary} cofibration), and prove that univalence implies
it; and
we isolate a \emph{weak} Glue structure (boundary agreement only by
inner equalities in the universe, an unglue consisting of mere maps, constant
background family, special boundary data, and \emph{no} cofibrancy
assumption) and prove that it implies univalence. Every notion of
gluing that lies between the two is then equivalent to univalence. The
forward direction works along an arbitrary cofibration and makes no
mention of an interval. The converse direction needs a
single cofibration rich enough to \emph{detect} equivalences; the
structure required for this, the \emph{path interval}, is set up in
\cref{sec:path-interval}, between the two directions.

\subsection{From univalence to gluing: the strong direction}

\begin{longversion}
As a baseline, we first record what univalence yields with no 2LTT
assumptions at all: it makes the pointwise formula of
\cref{def:pointwise-glue} a Glue object up to inner equality in the universe.

\begin{remark}[Pointwise Glue from univalence]\label{rem:pointwise-glue-from-UA}
	Let \((\varphi,T,e,A)\) be input data as in \cref{def:pointwise-glue},
	assume \(\UU\) is closed under the \(\Sigma\)- and \(\Pi\)-types
	displayed there, and write
	\(G \colonequiv \mathsf{Glue}_{\mathrm{pt}}(\varphi,T,e,A)\).
	Entirely within HoTT: the fiber of \(\mathsf{unglue}\) over \(a:A\) is
	\(\prod_{p:\varphi}\hfib{e_p}(a)\), a product of contractible types
	(each \(e_p\) is an equivalence), so \(\mathsf{unglue}\) is an equivalence;
	and for \(p:\varphi\), the map \(\rho_p(a,h)\colonequiv\pi_1(h(p))\)
	satisfies \(e_p\circ\rho_p\sim\mathsf{unglue}\), hence is a boundary
	equivalence \(\rho_p:G\simeq T(p)\) by two-out-of-three. If \(\UU\) is
	moreover univalent, each \(\rho_p\) yields an inner equality
	\(\mathsf{ua}(\rho_p):G=_\UU T(p)\), so \(G\) restricts to \(T\) along
	\(\varphi\) up to equality in the universe. Note where univalence enters:
	only to upgrade the boundary equivalence to an inner equality.

	None of this gives a judgmental or strict boundary \(G\equiv T(p)\), a
	\(\mathsf{glue}\) constructor with computation rules, or Kan composition.
	Like the weak structure of \cref{lem:ua-gives-weak-glue}, the pointwise
	Glue is a \emph{baseline} (what univalence yields with no cofibrancy
	assumptions), there along the boundary shape \((I,0,1)\), here along an
	arbitrary proposition \(\varphi\). When \(\varphi\) is a cofibrant strict
	proposition (both \(\szero\to[\varphi]\) and \([\varphi]\to\one\)
	cofibrations),
	\cref{thm:ua-implies-strong-glue} below, applied to the cofibration
	\([\varphi]\hookrightarrow\one\), strictly subsumes this baseline,
	producing a glued type with \emph{strict} boundary, an unglue
	equivalence with strict boundary agreement, and contractible uniqueness
	of all the data.
\end{remark}
\end{longversion}

Throughout, for $\beta:G\circ i\seq T$ we write
$e^\beta_\phi:G(i\phi)\simeq A(i\phi)$ for $e_\phi$ with its source identified
with $G(i\phi)$ along $\beta_\phi$.

\begin{lemma}[\flink{Lemma-5-1} Glue data forms a fibrant type]\label{lem:glue-data-fibrant}
	Let $i:\Phi\hookrightarrow\Gamma$ be a cofibration, $\UU$ a fibrant
	universe, and $(A,T,e)$ an input triple, with
	$\mathsf{GlueData}_{i,\UU}(A,T,e)$ as in \cref{def:glue-structure}. Let $W$
	be the fibrant family
	over $\Gamma$ given by $W(\gamma)\colonequiv\Sigma(X:\UU).\,(X\simeq A(\gamma))$,
	and let $w_0\colonequiv\lambda\phi.(T\phi,\,e_\phi):\prod_\Phi W\circ i$.
	Then the restriction map $i^*_W:\prod_\Gamma W\to\prod_\Phi W\circ i$ is a
	fibration, and
	\[
		\mathsf{GlueData}_{i,\UU}(A,T,e)\;\siso\; \sfib{(i^*_W)}(w_0),
	\]
	its strict fiber over $w_0$; in the notation of \cref{def:ext-formal},
	this is the extension type $\Ext{i}{\Pi_\Gamma W}{w_0}$: the type of Glue data \emph{is}
	an extension type, namely that of the universal family $W$. In particular,
	$\mathsf{GlueData}_{i,\UU}(A,T,e)$ is fibrant. Univalence is not used.
\end{lemma}
\begin{proof}
	Unfolding the notation of \cref{def:glue-structure},
	\[
		\mathsf{GlueData}_{i,\UU}(A,T,e)
		\;=\;
		\Sigma(G:\Gamma\to\UU).\,
		\Sigma(\beta:G\circ i\seq T).\,
		\Sigma\bigl(u:\textstyle\prod_{\gamma:\Gamma}G(\gamma)\simeq A(\gamma)\bigr).\,
		\bigl(u\circ i\seq e^\beta\bigr).
	\]
	By strict $\Sigma$/$\Pi$ distributivity,
	$\Sigma(G:\Gamma\to\UU).\prod_\gamma(G\gamma\simeq A\gamma)\siso\prod_\Gamma W$. Under
	this strict isomorphism, decomposing the strict equality of pairs (the
	second component transported along the first, which is where $e^\beta$
	arises) and using strict function extensionality, the pair of boundary
	conditions $(\beta,c)$ corresponds exactly to the single condition
	$w\circ i\seq w_0$. This yields the displayed strict isomorphism. Since $i$
	is a cofibration and $W$ is fibrant, $i^*_W$ is a fibration
	\citep[Lem.~3.18]{2LTT}; strict fibers of fibrations are fibrant
	\citep[Lem.~3.9]{2LTT}, and fibrancy transfers along strict isomorphisms.
\end{proof}

\begin{definition}[\flink{Definition-5-2} Contractible Glue data]\label{def:strong-glue}
	We say that $(i,\UU)$ has \emph{contractible Glue data} if
	$\mathsf{GlueData}_{i,\UU}(A,T,e)$ is contractible for every input triple
	$(A,T,e)$, a meaningful demand by \cref{lem:glue-data-fibrant}.
\end{definition}

\begin{remark}[Reading \cref{def:strong-glue}]\label{rem:strong-glue-eep}
	The center of contraction is exactly a Glue structure in
	the sense of \cref{def:glue-structure}, a glued family with strict type
	boundary \emph{and} strict unglue boundary
	(\cref{lem:glue-strength-chain}(1) below). Beyond this existence content,
	contractibility asserts \emph{uniqueness}: the glued type, its strict
	boundary, the unglue equivalence and its strict boundary agreement are all
	contractibly unique.
	Equivalently, writing
	$R_A:\Sigma(G:\Gamma\to\UU).\prod_\gamma(G\gamma\simeq A\gamma)\;\longrightarrow\;\Sigma(T:\Phi\to\UU).\prod_\phi(T\phi\simeq A(i\phi))$
	for the restriction map $(G,u)\mapsto(G\circ i,u\circ i)$, having
	contractible Glue data says that every strict fiber of $R_A$ is
	contractible.
	This is a 2LTT-internal, strict-boundary analogue of Sattler's equivalence extension property~\citep[Sec.~5]{Sattler2017}, strengthened by contractibility.
	(Under the strict
	$\Sigma$/$\Pi$-distributivity isomorphisms of \cref{lem:glue-data-fibrant},
	$R_A$ is identified with the restriction fibration $i^*_W$.)
\end{remark}

\begin{theorem}[\flink{Theorem-5-4} Univalence implies contractible Glue data]\label{thm:ua-implies-strong-glue}
	Let $i:\Phi\hookrightarrow\Gamma$ be a cofibration and let $\UU$ be a
	fibrant univalent universe. Then $(i,\UU)$ has contractible Glue data. No
	cofibrancy of $\Phi$ or $\Gamma$ is needed.
\end{theorem}
\begin{proof}
	Fix $(A,T,e)$, and let $W$, $w_0$ and the restriction fibration
	$i^*_W:\prod_\Gamma W\to\prod_\Phi W\circ i$ be as in
	\cref{lem:glue-data-fibrant}, so that
	$\mathsf{GlueData}_{i,\UU}(A,T,e)\siso \sfib{(i^*_W)}(w_0)$.

	\smallskip\noindent\textbf{Step 1: univalence makes $W$ pointwise contractible.}
	By univalence, $\mathsf{idtoeqv}_{X,A\gamma}:(X=A\gamma)\to(X\simeq A\gamma)$ is an
	equivalence for every $X$; totalizing over $X$ \citep[\S4.7]{HoTTBook} gives
	$W(\gamma)\simeq\Sigma(X:\UU).(X=A\gamma)$, a singleton, so each $W(\gamma)$ is
	contractible.

	\smallskip\noindent\textbf{Step 2: conclusion, by relative funext.}
	By Step~1, $W$ is a family of trivially fibrant types, so restriction along
	the cofibration $i$ is a \emph{trivial} fibration \citep[Lem.~3.18]{2LTT}
	in a universe large enough to contain $W$, and its strict fiber
	over $w_0$ is trivially fibrant \citep[Def.~3.7]{2LTT}: the extension type
	$\Ext{i}{\Pi_\Gamma W}{w_0}\siso\mathsf{GlueData}_{i,\UU}(A,T,e)$
	(\cref{lem:glue-data-fibrant}) is contractible. As in
	\cref{lem:rel-funext-holds}, this uses the trivial-fibration half of the
	cofibration property, and no cofibrancy of $\Phi$ or $\Gamma$.
	(When $\Phi$ and $\Gamma$ \emph{are} cofibrant, one can argue instead via
	\cref{lem:weak=strict-dep}: both $\prod_\Gamma W$ and $\prod_\Phi W\circ i$
	are contractible by \cref{lem:rel-funext-holds}, and
	$\Ext{i}{\Pi_\Gamma W}{w_0}$ is equivalent to the homotopy fiber of $i^*_W$ over
	$w_0$, a homotopy fiber of a map between contractible types. In the Agda
	formalization this is used for cofibrant shapes.)
\end{proof}

In terms of the Glue \emph{rules} in homotopy form
(\cref{lem:homotopy-vs-semantic}),
the contractibility of the Glue data yields the following.

\begin{corollary}[\flink{Corollary-5-5} Univalence gives the homotopy form]\label{cor:ua-gives-homotopy-form}
	Let $\UU$ be a fibrant univalent universe and $i:\Phi\hookrightarrow\Gamma$
	a cofibration with $\Phi$ (hence $\Gamma$) cofibrant. Then $(i,\UU)$
	satisfies the Glue rules in homotopy form
	(\cref{def:homotopy-glue-rules}), with the underlying Glue data
	contractibly unique
	(\cref{thm:ua-implies-strong-glue}), and hence, by
	\cref{lem:constructor-package}(2), the full constructor package up to
	homotopy, with \rulename{term-bdry} strict. Only the strict $\eta$- and
	$\beta$-laws are missing, and by \cref{lem:constructor-package}(1) they
	are exactly the demand that the equivalences $\Theta$ be strictified to
	strict isomorphisms, a realignment demand.
\end{corollary}
\begin{proof}
	\Cref{thm:ua-implies-strong-glue} makes
	$\mathsf{GlueData}_{i,\UU}(A,T,e)$ contractible for every input triple;
	its center of contraction is a Glue structure, which
	yields the homotopy form by \cref{lem:homotopy-vs-semantic}, and the
	contractibility gives the uniqueness.
\end{proof}

\subsection{The path interval}\label{sec:path-interval}

\Cref{thm:ua-implies-strong-glue} produces a Glue structure along
\emph{every} cofibration. A converse at that generality is hopeless: if the
only cofibrations available are trivial (for instance when no interesting shape
inclusions are postulated), a Glue structure carries no information about
$\UU$. To recover univalence we need at least one cofibration that is rich enough to
\emph{detect} equivalences (the inclusion of the two endpoints of an
interval), together with the structure that turns a line in the universe into
an inner equality and permits comparison with transport. We package this as a \emph{path
interval}.

\begin{definition}[\flink{Definition-5-6} Path interval]\label{def:path-interval}
	A \emph{path interval} consists of the following data.
	\begin{enumerate}
		\item An interval: a type $I$, not assumed fibrant, with two points $0, 1: I$, such that the induced map
		\[
		i_\partial:\partial I\longrightarrow I,
		\qquad \partial I \defeq \one +^{\mathrm{s}} \one
		\]
		is a cofibration.
		\item A line-to-identity structure: For every fibrant type \(Y\) (of arbitrary size), an operation
		\[
		\mathsf{lineToId}_Y:
		\prod_{g:I\to Y}\bigl(g(0)=_Y g(1)\bigr),
		\]
		natural in \(Y\). Naturality in \(Y\) means that, for all fibrant types \(Y,Z\), every function
		\(h:Y\to Z\), and every \(g:I\to Y\), there is an inner equality
		\[
		\mathsf{lineToId}_Z(h\circ g)
		=
		\mathsf{ap}_h\bigl(\mathsf{lineToId}_Y(g)\bigr).
		\]
	\end{enumerate}
\end{definition}

\noindent
\begin{minipage}[t]{0.755\textwidth}
\begin{remark}\label{rem:path-interval-diagram}
Writing $\mathsf{Id}_Y$ for $\Sigma (y_0, y_1 : Y). (y_0 = y_1)$ and $Y^I$ for $I \to Y$ to match the usual categorical presentation, the second condition of \cref{def:path-interval} means that, given the two solid fibrations in the diagram on the right, we have the dotted horizontal map.
	Some cubical models use $Y^I$ as their identity type, and one may expect that this horizontal map is an isomorphism.
	We do not require such an assumption. In CCHM with Swan's identity types~\citep[\S9.1]{CCHM2018}, the dotted map equips a line with the empty face marker, while forgetting the marker gives a retraction over $Y\times Y$. Neither the forgetful map nor the retraction law is part of \cref{def:path-interval}; they are not used below.
	Looking at cubical models, one might further expect the naturality condition to hold strictly; we do not ask for this as no argument in this paper requires it.
\end{remark}
\end{minipage}\hfill
\begin{minipage}[t]{0.215\textwidth}
	\vspace{0.6\baselineskip}
	\centering
	$\begin{tikzcd}[column sep=0.1em, row sep=4em]
		Y^I  \ar[rr,dotted] \ar[dr,->>] && \mathsf{Id}_Y \ar[dl,->>]\\
		& Y \times Y
	\end{tikzcd}$
\end{minipage}

\begin{lemma}[\flink{Lemma-5-8-coe} Coercion laws]\label{lem:coercion-laws}
	Let $I$ be a path interval and $\UU$ a fibrant universe. For a family
	$G:I\to\UU$ define
	\[
		\mathsf{coe}^G\;\colonequiv\;\pi_1\bigl(\mathsf{idtoeqv}(\mathsf{lineToId}(G))\bigr)\;:\;G(0)\to G(1).
	\]
	Then:
	\begin{enumerate}[noitemsep]
		\item[(c1)] by definition, $\mathsf{coe}^G$ is the underlying function of
		$\mathsf{idtoeqv}(\mathsf{lineToId}(G)):G(0)\simeq G(1)$;
		\item[(c2)] \emph{(naturality)} for families $G,H:I\to\UU$ and a
		fiberwise map $\theta:\prod_{t:I}\bigl(G(t)\to H(t)\bigr)$, one has
		$\theta_1\circ\mathsf{coe}^G\sim\mathsf{coe}^H\circ\theta_0$;
		\item[(c3)] \emph{(degeneracy)} $\mathsf{coe}^{\lambda\_.\,Y}\sim\id_Y$
		for every $Y:\UU$.
	\end{enumerate}
\end{lemma}
\begin{proof}
	(c3): The constant family satisfies $\lambda\_.\,Y\seq k\circ\,!$ for
	$!\colonequiv\lambda\_.\,\star:I\to\one$ and $k\colonequiv\lambda\_.\,Y:\one\to\UU$.
	The inner equality $\mathsf{lineToId}(!)$ has type $\star=_\one\star$, which is
	contractible, so $\mathsf{lineToId}(!)=\mathsf{refl}$; by naturality,
	$\mathsf{lineToId}(\lambda\_.\,Y)=\mathrm{ap}_k(\mathsf{lineToId}(!))$, which
	is therefore equal to $\mathrm{ap}_k(\mathsf{refl})=\mathsf{refl}$.
	Hence $\mathsf{idtoeqv}(\mathsf{lineToId}(\lambda\_.\,Y))=\id$ and
	$\mathsf{coe}^{\lambda\_.Y}\sim\id_Y$.

	(c2): Let $\UU_{\to}\colonequiv\Sigma(X,Y:\UU).(X\to Y)$, a fibrant type,
	with projections $\mathsf{dom},\mathsf{cod}:\UU_\to\to\UU$. The fiberwise map
	$\theta$ packages as a single line
	$\xi\colonequiv\lambda t.(G\,t,\,H\,t,\,\theta_t):I\to\UU_\to$, with
	$\mathsf{dom}\circ\xi\seq G$ and $\mathsf{cod}\circ\xi\seq H$. Set
	$r\colonequiv\mathsf{lineToId}(\xi):\xi(0)=\xi(1)$; by naturality,
	$\mathrm{ap}_{\mathsf{dom}}(r)=\mathsf{lineToId}(G)$ and
	$\mathrm{ap}_{\mathsf{cod}}(r)=\mathsf{lineToId}(H)$. The characterization of
	identity types in the iterated $\Sigma$-type $\UU_\to$ and the transport formulas for
	function-type families and for the universe
	\citep[\S2.7, \S2.9, \S2.10]{HoTTBook} turn $r$ into a homotopy
	\[
		\mathsf{idtoeqv}(\mathrm{ap}_{\mathsf{cod}}\,r)\circ\theta_0
		\;\sim\;
		\theta_1\circ\mathsf{idtoeqv}(\mathrm{ap}_{\mathsf{dom}}\,r).
	\]
	Substituting $\mathsf{lineToId}(G)$ and $\mathsf{lineToId}(H)$ and taking underlying
	maps gives $\mathsf{coe}^H\circ\theta_0\sim\theta_1\circ\mathsf{coe}^G$.
\end{proof}

\noindent
The strong direction (\cref{thm:ua-implies-strong-glue}) applies to
$(i_\partial,\UU)$ whenever $\UU$ is a fibrant univalent universe. (The domain
$\partial I = \one+^{\mathrm{s}}\one$ is finite, hence cofibrant \citep[Lem.~3.25]{2LTT};
\cref{thm:ua-implies-strong-glue} no longer needs this, but it is used when
instantiating $(\mathsf{glue})$ at $i_\partial$ in
\cref{thm:conclusions}.) The simplification available when $I$ is
itself fibrant is discussed in \cref{rem:fibrant-interval}.

\begin{remark}[\flink{Remark-5-9} Simplification for a fibrant interval]\label{rem:fibrant-interval}
	If $I$ is itself fibrant, then, in addition to the endpoint cofibration, a single \emph{segment} $\mathsf{seg}:0=_I 1$ supplies the line-to-identity structure (clause~2) of \cref{def:path-interval}. The inner identity type of $I$ is then available, and one may take
	\[
		\mathsf{lineToId}(G)\colonequiv\mathsf{ap}_G(\mathsf{seg}),
		\qquad
		\mathsf{coe}^G\colonequiv\mathsf{transport}^G(\mathsf{seg}).
	\]
	The coercion laws of \cref{lem:coercion-laws} then hold for these definitions: they are the standard facts that $\mathsf{idtoeqv}(\mathsf{ap}_G \mathsf{seg})=\mathsf{transport}^G(\mathsf{seg})$, that fiberwise maps commute with transport, and that transport in a constant family is the identity \citep[\S2.3, \S2.10]{HoTTBook}; no separately supplied line-to-identity operation is needed.

	If we take \(I\) to be the HoTT interval, i.e.\ the higher inductive type that is generated by points
	\(0,1:I\) and a path \(\mathsf{seg}:0=_I1\), all conditions from \cref{def:path-interval} are satisfied apart from the condition that \(i_\partial:\one+^{\mathrm{s}}\one\rightarrow I\) is a cofibration.
	This does not appear to follow from the HIT rules and would have to be assumed separately.
\end{remark}

\subsection{From gluing to univalence: the weak direction}

The proof that gluing implies univalence uses only that the comparison map
$\mathsf{idtoeqv}$ admits a section.
This observation is due to Licata~\citep{Licata2016}, based on an earlier observation by Escard\'o~\citep{Escardo2014} that any retraction of an identity type is an equivalence:

\begin{lemma}[\flink{Lemma-5-10} Univalence from a section of $\mathsf{idtoeqv}$]\label{lem:ua-from-section}
	Let $\UU$ be a fibrant universe. Suppose that for all $A,B:\UU$ there is a
	map
	\[
		s_{A,B} : (A \simeq B) \longrightarrow (A =_\UU B)
	\]
	that is a section of $\mathsf{idtoeqv}_{A,B}$, i.e.\
	$\mathsf{idtoeqv}_{A,B}(s_{A,B}(g)) = g$ for every $g : A \simeq B$. Then
	$\UU$ is univalent.
\end{lemma}

\begin{proof}
	Fix $A:\UU$. Totalizing over $B$ turns the fiberwise maps
	$\mathsf{idtoeqv}_{A,-}$ and $s_{A,-}$ into
	\[
		\Theta : \sum_{B:\UU}(A =_\UU B) \longrightarrow \sum_{B:\UU}(A \simeq B),
		\qquad
		S : \sum_{B:\UU}(A \simeq B) \longrightarrow \sum_{B:\UU}(A =_\UU B),
	\]
	with $\Theta(B,p) \colonequiv (B,\mathsf{idtoeqv}(p))$ and
	$S(B,g) \colonequiv (B,s_{A,B}(g))$. The section hypothesis gives
	$\Theta(S(B,g)) = (B,\mathsf{idtoeqv}(s_{A,B}(g))) = (B,g)$, so
	$\Theta \circ S \sim \id$; that is, $\sum_{B}(A\simeq B)$ is a retract of
	$\sum_{B}(A =_\UU B)$. The latter is a singleton, hence contractible, and a
	retract of a contractible type is contractible \citep[\S3.11]{HoTTBook}.
	Thus $\sum_{B:\UU}(A\simeq B)$ is contractible for every $A:\UU$. This
	implies univalence: the map $\Theta$ now has contractible domain (a
	singleton) and contractible codomain, hence is an equivalence, and by the
	total-space characterization of fiberwise equivalences \citep[\S4.7]{HoTTBook}
	the fiberwise map $\mathsf{idtoeqv}_{A,B}$ is an equivalence for all $B$.
	As $A$ was arbitrary, $\UU$ is univalent.
\end{proof}

\begin{definition}[\flink{Definition-5-11} Weak Glue structure]\label{def:weak-glue}
	Let $I$ be a type with two points $0,1:I$ and let $\UU$ be a fibrant
	universe. A \emph{weak Glue structure} for $(I,0,1,\UU)$ is an operation
	assigning to all $A,B:\UU$ and every equivalence $f:A\simeq B$ the
	following data:
	\begin{enumerate}[noitemsep]
		\item a family $G:I\to\UU$;
		\item \emph{inner equalities in the universe} $b_0:G(0)=_\UU A$ and $b_1:G(1)=_\UU B$;
		\item a family of \emph{mere maps} $u:\prod_{t:I}\bigl(G(t)\to B\bigr)$;
		\item homotopies $u_0\sim f\circ\mathsf{idtoeqv}(b_0)$ and
		      $u_1\sim \mathsf{idtoeqv}(b_1)$.
	\end{enumerate}
	No cofibrancy or fibrancy assumption on $I$ or on $0,1$ is made.
\end{definition}

Note that, in the above definitions, the boundary must consist of inner equalities.
If $b_0,b_1$ are weakened to equivalences, the structure is trivially inhabited
without univalence: take $G\colonequiv\lambda\_.\,B$,
$b_0\colonequiv f^{-1}$, $b_1\colonequiv\id_B$, and
$u\colonequiv\lambda t.\id_B$.
The requirement that the glued type
agree with the boundary data by (at least) an inner equality in the universe is the
irreducible core of gluing.

\begin{longversion}
\begin{remark}[The weakenings are sharp]\label{rem:weak-glue-sharp}
	Each ingredient of \cref{def:weak-glue} is as weak as possible.
	\begin{enumerate}[noitemsep]
		\item \emph{The unglue maps need not be equivalences}: invertibility at
		the endpoints is forced by the homotopies in (4), and the proof of
		\cref{thm:weak-glue-implies-ua} never inverts $u$.
		\item Only the \emph{constant} background family $\lambda\_.\,B$ and the
		special boundary data $(A,B,f,\id_B)$ are required, rather than
		arbitrary triples $(A,T,e)$.
	\end{enumerate}
	Conversely, the weak structure is easily inhabited \emph{from}
	univalence, for any type $I$ with two points and with no structure on
	$I$ at all: given $A,B:\UU$ and $f:A\simeq B$, take the constant family
	$G\colonequiv\lambda\_.\,B$, the boundary equalities
	$b_0\colonequiv\mathsf{ua}(f)^{-1}:B=_\UU A$ and
	$b_1\colonequiv\mathsf{refl}_B$, and the mere maps
	$u\colonequiv\lambda t.\,\id_B$; since
	$\mathsf{idtoeqv}(\mathsf{ua}(f))=f$ and $\mathsf{idtoeqv}$ carries
	inverses of inner equalities to inverse equivalences, the endpoint homotopies required
	by \cref{def:weak-glue}(4) follow.
\end{remark}
\end{longversion}

\begin{remark}[\flink{Remark-5-12} Univalence gives weak glue]\label{lem:ua-gives-weak-glue}
The weak structure is easily inhabited from
univalence, for any type $I$ with two points and with no structure on
$I$ at all: given $A,B:\UU$ and $f:A\simeq B$, take the constant family
$G\colonequiv\lambda\_.\,B$, the boundary equalities
$b_0\colonequiv\mathsf{ua}(f)^{-1}:B=_\UU A$ and
$b_1\colonequiv\mathsf{refl}_B$, and the mere maps
$u\colonequiv\lambda t.\,\id_B$; since
$\mathsf{idtoeqv}(\mathsf{ua}(f))=f$ and $\mathsf{idtoeqv}$ carries
inverses of inner equalities to inverse equivalences, the endpoint homotopies required
by \cref{def:weak-glue}(4) follow.
\end{remark}

By \cref{rem:fibrant-interval}, if \(I\) is fibrant and carries a segment
\(\sigma:0=_I1\), then \(I\) carries a line-to-identity structure. Hence
\cref{thm:weak-glue-implies-ua,lem:ua-gives-weak-glue} show that univalence
and the existence of a weak Glue structure for \(I\) imply one another.
Neither direction requires \(i_\partial:\partial I\to I\) to be a
cofibration. Thus this conclusion does not assert that \(I\) is a path
interval in the sense of \cref{def:path-interval}.

\begin{lemma}[\flink{Lemma-5-13-i} Chain of strength]\label{lem:glue-strength-chain}
	Let $\UU$ be a fibrant universe.
	\begin{enumerate}[noitemsep]
		\item For every cofibration $i$, contractible Glue data for
		$(i,\UU)$ yield a Glue structure for $(i,\UU)$.
		\item Let $I$ be a type with $0,1:I$ such that
		$i_\partial:\partial I\hookrightarrow I$ is a cofibration. Then
		a Glue structure for $(i_\partial,\UU)$ yields a weak Glue structure
		for $(I,0,1,\UU)$. The same conclusion holds for the variant in which
		the strict coherence of \cref{def:glue-structure} is replaced by a
		homotopy between the underlying functions.
	\end{enumerate}
\end{lemma}
\begin{proof}
	(1) Take the center of contraction. (2) The proof uses the strict
	coherence only after projecting to the underlying functions and converting
	it into a homotopy. We may therefore assume these boundary homotopies
	directly.
	Given $A,B:\UU$ and $f:A\simeq B$, apply the given structure to the
	constant family $\bar A\colonequiv\lambda\_.\,B$, the
	copairing $T\colonequiv[A,B]:\partial I\to\UU$, and
	$e\colonequiv[f,\id_B]$. We obtain $G:I\to\UU$ with
	$\beta:G\circ i_\partial\seq T$, an unglue family
	$u:\prod_t\bigl(G(t)\simeq B\bigr)$. Write $u'_t\colonequiv\pi_1(u_t)$
	for its underlying functions. The assumed boundary homotopies are
	$u'_{i_\partial\phi}\sim\pi_1(e_\phi)\circ c_{\beta_\phi}$, where
	$c_{\beta_\phi}:G(i_\partial\phi)\to T(\phi)$ is the strict transport along
	$\beta_\phi$. Set $b_0\colonequiv\iota(\beta_{\mathsf{inl}}):G(0)=_\UU A$
	and $b_1\colonequiv\iota(\beta_{\mathsf{inr}}):G(1)=_\UU B$. By strict equality elimination, for
	any strict equality $\beta$ the underlying function of
	$\mathsf{idtoeqv}(\iota(\beta))$ is homotopic to the strict transport $c_\beta$ (at
	$\mathsf{refl}^{\mathrm{s}}$, the strict transport is the identity and
	$\mathsf{idtoeqv}(\mathsf{refl})$ is homotopic to it by the propositional
	$J$-$\beta$ rule; with weak $J$-$\beta$, only this homotopy is available, and
	only it is used). Hence
	$u'_0 \sim f\circ c_{\beta_{\mathsf{inl}}} \sim f\circ\mathsf{idtoeqv}(b_0)$
	and
	$u'_1 \sim \id_B\circ c_{\beta_{\mathsf{inr}}} \sim \mathsf{idtoeqv}(b_1)$,
	which is the weak structure.
\end{proof}

\begin{theorem}[\flink{Theorem-5-14} A weak Glue structure implies
	univalence]\label{thm:weak-glue-implies-ua}
	Let $\UU$ be a fibrant universe and let $(I,0,1)$ have a line-to-identity structure (i.e., satisfy clause~2 of
	\cref{def:path-interval}). If there is a weak Glue structure for
	$(I,0,1,\UU)$, then $\UU$ is univalent.
\end{theorem}
\begin{proof}
	By \cref{lem:ua-from-section} it suffices to construct, for all $A,B:\UU$, a
	section of $\mathsf{idtoeqv}_{A,B}$. Given $f:A\simeq B$, let
	$(G,b_0,b_1,u)$ be the data provided by the weak Glue structure and set
	\[
		s(f)\;\colonequiv\;b_0^{-1}\cdot\mathsf{lineToId}(G)\cdot b_1
		\;:\;A=_\UU B ,
	\]
	using $\mathsf{lineToId}(G):G(0)=_\UU G(1)$ from the line-to-identity structure at the
	fibrant type $\UU$. By the groupoid laws for $\mathsf{idtoeqv}$ and
	\cref{lem:coercion-laws}(c1), the underlying map of $\mathsf{idtoeqv}(s(f))$
	is $\mathsf{idtoeqv}(b_1)\circ\mathsf{coe}^G\circ\mathsf{idtoeqv}(b_0)^{-1}$.
	Now compute, using the homotopies of \cref{def:weak-glue}(4), naturality
	\cref{lem:coercion-laws}(c2) applied to the fiberwise map $u$ from $G$ to
	the constant family $\lambda\_.\,B$, and degeneracy (c3):
	\[
		\mathsf{idtoeqv}(b_1)\circ\mathsf{coe}^G
		\;\sim\;u_1\circ\mathsf{coe}^G
		\;\sim\;\mathsf{coe}^{\lambda\_.B}\circ u_0
		\;\sim\;u_0
		\;\sim\;f\circ\mathsf{idtoeqv}(b_0).
	\]
	Hence $\mathsf{idtoeqv}(s(f))\sim f$ pointwise. Being an equivalence is a
	proposition, so $\mathsf{idtoeqv}(s(f))=f$ as equivalences, $s$ is a section of
	$\mathsf{idtoeqv}_{A,B}$, and \cref{lem:ua-from-section} yields univalence.
\end{proof}

\subsection{The equivalence cycle}

We put the above results together:

\begin{theorem}[\flink{Theorem-5-15-fibrant} Univalence is equivalent to gluing, strong form]\label{thm:glue-sandwich}
	Let $I$ be a path interval (\cref{def:path-interval}) and let
	$\UU$ be a fibrant universe. The following are equivalent:
	\begin{enumerate}[noitemsep]
		\item $\UU$ is univalent;
		\item for \emph{every} cofibration $i:\Phi\hookrightarrow\Gamma$,
		$(i,\UU)$ has contractible Glue data;
		\item $(i_\partial:\partial I\hookrightarrow I,\ \UU)$ has contractible
		Glue data;
		\item there is a Glue structure for $(i_\partial,\UU)$;
		\item there is a weak Glue structure for $(I,0,1,\UU)$.
	\end{enumerate}
\end{theorem}
\begin{proof}
	$(1)\Rightarrow(2)$ is \cref{thm:ua-implies-strong-glue}.
	$(2)\Rightarrow(3)$ is a special case.
	$(3)\Rightarrow(4)$ is \cref{lem:glue-strength-chain}(1) and
	$(4)\Rightarrow(5)$ is \cref{lem:glue-strength-chain}(2).
	$(5)\Rightarrow(1)$ is \cref{thm:weak-glue-implies-ua}.
\end{proof}

\begin{longversion}
  \begin{remark}\label{rem:glue-sandwich-consequences}
	In \cref{thm:glue-sandwich}, the \emph{contractible Glue data} of
	\cref{def:strong-glue}, the \emph{Glue structure} of
	\cref{def:glue-structure}, and the \emph{weak Glue structure} of
	\cref{def:weak-glue} are semantic packages. Thus, the equivalence cycle
	relates univalence to these semantic gluing conditions. It does not assert
	that univalence is equivalent to the availability of the primitive CCHM
	\texttt{Glue} type former together with its \texttt{glue} constructor,
	judgmental computation rules, and Kan composition
	\citep[Fig.~4, \S\S6--7]{CCHM2018}.
	
	The fragment $(1)\Leftrightarrow(4)$ of \cref{thm:glue-sandwich} is the
	basic equivalence of univalence with gluing along $i_\partial$.
	On the forward side, univalence provides contractible Glue data along
	every cofibration, rather than merely a selected Glue object. For each
	input triple, the type of glued families equipped with a strict type
	boundary and a strictly coherent unglue equivalence is contractible.
	(Existence alone, with boundary coherence only up to homotopy, can also
	be obtained in a single strictification step from \cref{lem:rs410}.)
	
	On the converse side, univalence is already forced by the weak structure
	of \cref{def:weak-glue}. This structure requires only boundary agreement
	by inner equalities in the universe, mere unglue maps, a constant
	background, and the special boundary $(A,B,f,\id_B)$. It does not require
	$i_\partial$ to be a cofibration. Of the path-interval data in
	\cref{def:path-interval}, the proof uses only the line-to-identity clause,
	via \cref{lem:coercion-laws}, at the three types $\UU$,
	$\Sigma(X,Y{:}\UU).(X\to Y)$, and $\one$.
	Conversely, univalence directly supplies a weak Glue structure for every
	bipointed type $(I,0,1)$, with no further hypothesis on $I$
	(\cref{lem:ua-gives-weak-glue}).
	
	Under the hypotheses of \cref{thm:glue-sandwich}, any gluing condition
	that follows from contractible Glue data and implies weak Glue along
	$(I,0,1)$ is equivalent to univalence. This applies, for example, to the
	variant of \cref{def:glue-structure} in which strict coherence is replaced
	by a homotopy between the underlying functions:
	\cref{lem:glue-strength-chain}(2) gives its implication to weak Glue, while
	the implication from contractible Glue data factors through an ordinary
	Glue structure. The formalization proves the cycle for this variant as
	well. The same conclusion applies to a realignment-style variant once
	these two comparison implications have been established.
\end{remark}
\end{longversion}

\section{Conclusions: Toward an equivalence of type theories}\label{sec:conservativity}

\paragraph{The question.} The question about the connection between cubical type theories and ``book HoTT'' is as old as cubical type theory, and arguably among the most important open questions of homotopy type theory.
What researchers studying the field hope is a statement of the form:
\begin{conjecture}[folklore conjecture of homotopy type theory]\label{conj:is-this-possible}
	Internal theorems proved in a cubical type theory such as cubical Agda~\citep{CubicalAgda} can also be proved in the version of HoTT developed in the HoTT book~\citep{HoTTBook}, as approximated e.g.\ by standard Agda with postulates as in the \texttt{HoTT-Agda} library~\citep{swan:hott-in:agda}.
	Concretely, various formulations of homotopy type theory and cubical type theories are Morita equivalent or conservative extensions of each other.
\end{conjecture}

The Morita equivalence formulation of \cref{conj:is-this-possible} is due to Isaev
\citep{Isaev2018Morita}, whose motivating examples include replacing
judgmental computation rules by propositional ones (proved only in simple
cases). Uemura poses exactly our question: interpreting book HoTT in
cubical type theory is expected to be an equivalence up to homotopy, never
one of syntactic theories \citep[\S8]{Uemura2023}; cf.\ the $\infty$-type
theories of \citet{NguyenUemura2025}.
Hofmann's conservativity of extensional over intensional type theory
\citep{Hofmann1995}, made constructive by Oury \citep{Oury2005}, completed and
formalized by Winterhalter et al.~\citep{WST2019,Winterhalter2020}, and by
now a Morita equivalence \citep{KapulkinLi2025Revisited}, is the principle that strictness can be inert; its UIP-free analogues are the target of Bocquet's coherence
program \citep{Bocquet2020,Bocquet2022,Bocquet2023,BKS2023}, highly relevant for the questions at hand.
Semantically, BCH, De Morgan,
and minimal cartesian cubical sets present homotopy theories
\emph{different} from spaces \citep{Sattler2018Talk,ACCRS2026}, while the
equivariant cartesian model, where the connection is most plausible,
presents spaces \citep{ACCRS2026,CavalloSattler2025,CavalloSattler2026}; the
cleanest positive evidence is homotopy canonicity
\citep{CHS2022,KapulkinSattler2019,BocquetRezk2023}. In every one of these
senses, the full question is (at the time of writing) open, to the knowledge of the current author.
Our developments in the previous sections are not sufficient to settle any non-artificial version of this conjecture, but they may suggest one route together with partial results.


\paragraph{Representing structural extensions as axiomatic extensions.}
The strategy suggested by the results of this paper is an instance of a general approach that we presented at TYPES 2025~\citep{KdJ_representing}:
use 2LTT to represent \emph{structural} extensions of type theories as
\emph{axiomatic} ones. An axiomatic extension adds constants and axioms
to a theory; for example, univalence is added to MLTT to obtain HoTT as in the book \citep{HoTTBook}. It
is comparatively harmless: intuition, results, and proof assistants
transfer directly. A structural extension instead adds new judgments,
context layers, or judgmental equalities. For example,
cubical type theories extend MLTT with an
interval, a face lattice, and judgmental laws of Glue and
composition, and simplicial type theory \citep{RS2017} extends HoTT with
shape context layers.
Structural extensions may require new implementations, for example cubical Agda~\citep{CubicalAgda} or
Rzk~\citep{rzk}, and new meta-theory. The extension of HoTT to
2LTT is itself structural, but harmless in the sense discussed in the introduction;
most importantly, thanks to the known
conservativity results \citep{2LTT,andras-staging,BocquetThesis}, we know that the fibrant fragment
proves the same internal theorems.
The point of the mentioned general approach is that a
structural extension of HoTT may become an axiomatic
extension within 2LTT, where it can be studied and compared with existing
tools, for example using Agda's \texttt{--two-level} flag
\citep{agdareadthedocs}, implemented without a custom proof assistant.

\paragraph{Representing a cubical type theory and book HoTT in 2LTT\@.}
A first subtlety is that, a priori, most cubical type theories are not actual extensions of book HoTT (i.e., do not model book HoTT) because the latter comes with a strict $J_\beta$ rule for the eliminator of the identity type.
In contrast, in cubical type theories, the Path-type
eliminator generally satisfies $J_\beta$ only weakly.
Swan's identity types \citep{Swan2016}, adopted by CCHM~\citep[\S9.1]{CCHM2018}, fix this mismatch, by additionally ``marking'' where a path is constant.
The effect is that the total space of Swan's identity type at type $Y$, i.e.\ $\Sigma(y_0, y_1: Y), y_0 =_\mathsf{swan} y_1$, is generally not isomorphic to function types $I \to Y$.
This motivates our seemingly weak formulation of path intervals (\cref{def:path-interval}); Swan's identity types model this definition by sending a cubical line to the same line with the empty face marker.%
\footnote{The endpoints are immediate. Naturality is not, because the operation just described is the \emph{direct} action of $h$ on marked lines, whereas $\mathrm{ap}_h$ is defined by identity elimination; the two agree by identity elimination on the marked line, both sending $\mathsf{refl}$ (the degenerate line carrying the total face) to $\mathsf{refl}$, using that Swan's $J$ computes strictly on $\mathsf{refl}$. Applying the resulting inner equality at the line with the empty marker is clause~2 of \cref{def:path-interval}.}
This allows us to view a version of cubical type theory as an extension of book HoTT.%
\footnote{Note that the type theory presented in \cite{HoTTBook} (and also Rijke's book \cite{Rijke_2025}) lacks the judgmental $\eta$-rule for the unit type and $\Sigma$-types, which is assumed here. Proof assistants generally support $\eta$ for unit and record types, and $\Sigma$ is often implemented as a record type.
What we call ``book HoTT'' here is thus closer to what Agda and other proof assistants implement than to the \emph{actual} theories of the book(s). Conservativity of these judgmental $\eta$-rules should be studied and established independently of the comparison of book HoTT and cubical type theories. Some results along these lines can be found in \cite{Isaev2018Morita} and \cite{Bocquet2020,Bocquet2023}.}

\newcommand{\axUA}{\mathsf{(ua)}}
\newcommand{\axInt}{\mathsf{(int)}}
\newcommand{\axGlue}{\mathsf{(glue)}}

\begin{definition}[type theories represented within 2LTT]\label{def:repres-tts}
	Let $\bbMLTT$ be the two-level type theory of
	\citet{2LTT}, instantiated as follows.
	\begin{enumerate}[noitemsep]
		\item \emph{Inner level:} MLTT with $\Pi$, $\Sigma$, binary sums,
		empty and unit type, natural numbers (and, more generally, $W$-types), identity types, and a
		cumulative hierarchy of universes $\UU_0:\UU_1:\cdots$.
		The $\beta$-rules hold judgmentally and, in addition, $\Pi$, $\Sigma$, and unit type have judgmental $\eta$.
		Function extensionality is assumed as a postulate.
		\item \emph{Outer level:} The outer level is extensional MLTT (i.e.\ with the equality reflection rule that identifies the equality type and judgmental equality) with the analogous components. Note that UIP and function extensionality are derivable in extensional MLTT.
		\item For the conversion function between inner and outer level, we assume axioms (T1) and (T2) of \cite{2LTT} (i.e., the conversion preserves $\Sigma$, $\Pi$, and unit).
	\end{enumerate}
	Within this setting, we consider the following additional assumptions for all inner universes $\UU$, each phrased as a postulate on either the inner or the outer level:
	\begin{itemize}[noitemsep]
		\item[$\axUA$] the universe $\UU$ is univalent, phrased as a postulate on the inner level;
		\item[$\axGlue$] for every cofibration $i:\Phi\cofib\Gamma$ with $\Phi$ cofibrant, there is a Glue structure for $(i,\UU)$ (\cref{def:glue-structure}), phrased as a postulate on the outer level;
		\item[$\axInt$] there is a path interval $(I,i_\partial)$ (\cref{def:path-interval}), again phrased as a postulate on the outer level.
	\end{itemize}
	Finally, we denote by $\bbHoTT$ the two-level type theory $\bbMLTT$ extended with the univalence axiom, written as $\bbHoTT \defeq \bbMLTT + \axUA$.
	Similarly, we write $\bbCub \defeq \bbMLTT + \axGlue + \axInt$.
\end{definition}
Note that, while in all earlier sections we only assumed the outer level to satisfy UIP and function extensionality, \cref{def:repres-tts} assumes it to have equality reflection.
This makes little mathematical difference on paper due to the conservativity of extensional MLTT over intensional MLTT with UIP and function extensionality~\cite{Hofmann1995, Oury2005,Winterhalter2020}; the conservativity result is made precise in this setting in~\cite[Prop.~2.19]{2LTT}.
This makes the fibrant fragment of $\bbHoTT$ a good approximation of book HoTT, while $\bbCub$ is still merely a ``toy version'' of a cubical type theory.
The assumption of equality reflection cannot (easily) be replicated in Agda, but all results stay valid in the stronger setting, and it leads to a cleaner formulation as the distinction between strict and judgmental equality disappears.
In this situation \cref{thm:glue-sandwich} shows that, over $\bbMLTT + \axInt$, the axioms $\axUA$ and $\axGlue$ are inter-derivable; in other words:

\begin{theorem}\label{thm:conclusions}
	The internal theories $\bbHoTT + \axInt$ and $\bbCub$ are mutually interpretable by translations that are the identity on the common signature, and they prove the same judgments in the language of $\bbMLTT + \axInt$.
	\qed
\end{theorem}

\cref{thm:conclusions} is at best a ``toy comparison''
and merely a first step toward the open problem of connecting book HoTT
and cubical type theories. We nevertheless see it as a proof of concept
for the method of this paper: 
Thanks to the known conservativity results of 2LTT, it is reasonable to hope that the representations of type theories in 2LTT are faithful. Once this is done, the core of the comparison, namely the
interchangeability of gluing and univalence, can be phrased internally, proved by internal reasoning and machine-checked in Agda.
Of course, there are more structural differences between the involved type theories than the discussed ones, and the question whether the unwanted assumption $\axInt$ on the $\bbHoTT$ side can be discharged remains open for now.
These questions connect the problem to the active work on conservativity by Bocquet~\citep{rafael_framework}.
Extension types, which the setting of the current paper provides for free, supply a calculus that may make such comparisons manageable, and our hope is that they prove useful for other structural extensions of type theories.

\section*{Acknowledgments}
I thank Joshua Chen, Sti\'ephen Pradal, and Tom de Jong for listening to some of the ideas of this paper and offering feedback during a series of seminar meetings in 2023/24.
Further, I thank Christian Sattler and Jonas H\"ofer for a discussion on cofibrations, Glue, and especially the pointwise formulas (cf.\ the end of \cref{sec:gluing}) at the Institut Mittag-Leffler in 2026.

\bibliographystyle{plainnat}
\bibliography{references}
\end{document}